\documentclass[12pt]{article}
\usepackage[T1]{fontenc}
\usepackage[latin9]{inputenc}
\usepackage{geometry}
\usepackage{mathtools}
\usepackage{amsmath,amsthm,amssymb,amstext,amsbsy} 
\usepackage{latexsym,graphicx}
\usepackage{subfig}
\usepackage[usenames,dvipsnames]{xcolor}
\usepackage{thmtools,thm-restate}
\usepackage{comment}

\usepackage{libertine}
\usepackage[libertine]{newtxmath}
\usepackage{algorithm}
\usepackage{algpseudocode}

\graphicspath{{./figs/}} 

\definecolor{mygreen}{RGB}{80,180,0} 
\definecolor{b2}{RGB}{51,153,255}
\definecolor{myred}{RGB}{220,10,10}
\definecolor{mycy2}{RGB}{255,51,255}

\makeatletter
\newtheorem{theorem}{Theorem}[section]
\newtheorem{lemma}[theorem]{Lemma}
\newtheorem{definition}[theorem]{Definition}

\newtheorem{corollary}[theorem]{Corollary}

\newtheorem{fact}[theorem]{Fact}
\newtheorem{remark}[theorem]{Remark}

\newtheorem*{theorem*}{Theorem}

\usepackage{hyperref}
\hypersetup{
    colorlinks,
    allcolors=red
}

\usepackage{cleveref}

\@addtoreset{section}{part}

\makeatother

\begin{document}
\global\long\def\defeq{\stackrel{\mathrm{{\scriptscriptstyle def}}}{=}}
\global\long\def\norm#1{\left\Vert #1\right\Vert }
\global\long\def\R{\mathbb{R}}
 \global\long\def\Rn{\mathbb{R}^{n}}
\global\long\def\tr{\mathrm{Tr}}
\global\long\def\diag{\mathrm{diag}}
\global\long\def\cov{\mathrm{Cov}}
\global\long\def\E{\mathbb{E}}
\global\long\def\P{\mathbb{P}}
\global\long\def\Var{\mathrm{Var}}
\global\long\def\rank{\mathrm{rank}}
\global\long\def\lref#1{\text{Lem }\ltexref{#1}}
\global\long\def\lreff#1#2{\text{Lem }\ltexref{#1}.\ltexref{#1#2}}
\global\long\def\ltexref#1{\ref{lem:#1}}\global\long\def\ttag#1{\tag{#1}}
\global\long\def\cirt#1{\raisebox{.5pt}{\textcircled{\raisebox{-.9pt}{#1}}}}

\newcommand{\dd}{d}
\newcommand{\op}{\ensuremath{\mathrm{op}}\xspace}
\newcommand{\p}{\ensuremath{\mathbb{P}}\xspace}
\newcommand{\Gauss}{\ensuremath{\mathcal{N}}\xspace}
\newcommand{\spe}{\ensuremath{\mathrm{op}}\xspace}
\newcommand{\Def}{\ensuremath{\mathrm{def}}\xspace}
\newcommand{\dtv}{\ensuremath{d_\mathrm{TV}}\xspace}
\newcommand{\eps}{\epsilon}
\newcommand{\Ind}{\ensuremath{\mathbf{1}}}
\newcommand{\poly}{\mathsf{poly}}
\newcommand{\interior}{\mathsf{int}}
\newcommand{\vol}{\mathsf{vol}}
\newcommand{\wt}[1]{\widetilde{#1}}
\newcommand{\spn}{\mathsf{span}}
\newcommand{\SO}{\mathsf{SO}}
\newcommand{\EO}{\mathsf{EO}}
\renewcommand{\gcd}{\mathsf{gcd}}
\newcommand{\lcm}{\mathsf{lcm}}
\newcommand{\CVP}{\mathsf{CVP}}
\newcommand{\SVP}{\mathsf{SVP}}
\newcommand{\SC}{\mathsf{SC}}
\newcommand{\Cov}{\mathsf{Cov}}
\newcommand{\cg}{\mathsf{cg}}
\newcommand{\supp}{\mathsf{supp}}
\newcommand{\new}{\mathsf{new}}
\newcommand{\out}{\mathsf{out}}
\newcommand{\width}{\mathsf{width}}
\newcommand{\free}{\mathsf{free}}

\newcommand{\eat}[1]{}

\newcommand{\haotian}[1]{{\color{blue} \textbf{Haotian:} #1}}
\newcommand{\yintat}[1]{{\color{red} \textbf{Yin-Tat:} #1}}


\title{Minimizing Convex Functions with Rational Minimizers\thanks{To appear in the Journal of the ACM. This journal version simplifies and significantly strengthens the results in an earlier version of this paper which appeared in SODA 2021.}}

\author{Haotian Jiang \thanks{Paul G. Allen School of CSE, University of Washington, USA. \texttt{jhtdavid@cs.washington.edu}. Supported by NSF grants CCF-1749609, DMS-1839116 and DMS-2023166.}}

\date{}

\maketitle

\begin{abstract}
Given a separation oracle $\mathsf{SO}$ for a convex function $f$ defined on $\mathbb{R}^n$ that has an integral minimizer inside a box with radius $R$, we show how to find an exact minimizer of $f$ using at most 
\begin{itemize}
\item $O(n (n \log \log (n)/\log (n) + \log(R)))$ calls to $\mathsf{SO}$ and $\poly(n, \log(R))$ arithmetic operations, or 
\item $O(n \log(nR))$ calls to $\mathsf{SO}$ and $\exp(O(n)) \cdot \poly(\log(R))$ arithmetic operations. 
\end{itemize}

When the set of minimizers of $f$ has integral extreme points, our algorithm outputs an integral minimizer of $f$. This improves upon the previously best oracle complexity of $O(n^2 (n + \log(R)))$ for polynomial time algorithms and $O(n^2\log(nR))$ for exponential time algorithms obtained by [Gr\"otschel, Lov\'asz and Schrijver, Prog. Comb. Opt. 1984, Springer 1988] over thirty years ago. Our improvement on Gr\"otschel, Lov\'asz and Schrijver's result generalizes to the setting where the set of minimizers of $f$ is a rational polyhedron with bounded vertex complexity.

For the Submodular Function Minimization problem, our result immediately implies a strongly polynomial algorithm that makes at most $O(n^3 \log \log (n)/\log (n))$ calls to an evaluation oracle, and an exponential time algorithm that makes at most $O(n^2 \log(n))$ calls to an evaluation oracle. These improve upon the previously best $O(n^3 \log^2(n))$ oracle complexity for strongly polynomial algorithms given in [Lee, Sidford and Wong, FOCS 2015] and [Dadush, V\'egh and Zambelli, SODA 2018], and an exponential time algorithm with oracle complexity $O(n^3 \log(n))$ given in the former work. 

Our result is achieved via a reduction to the Shortest Vector Problem in lattices. We show how an approximately shortest vector of certain lattice can be used to effectively reduce the dimension of the problem. 
Our analysis of the oracle complexity is based on a potential function that captures simultaneously the size of the search set and the density of the lattice, which we analyze via tools from convex geometry and lattice theory.   
\end{abstract}

\newpage

\section{Introduction}

In this paper, we investigate the problem of minimizing a convex function $f$ on $\mathbb{R}^n$ accessed through a separation oracle $\SO$ \cite{gls81}. When queried with a point $x$, the oracle returns ``YES'' if $x$ minimizes $f$; otherwise, the oracle returns a hyperplane that separates $x$ from the minimizer of $f$.   
An algorithm is said to be {\em strongly polynomial}~\cite{gls88} for such a problem if it makes $\poly(n)$ calls to $\SO$, uses $\poly(n)$ arithmetic operations, and the size of numbers occurring during the algorithm is polynomially bounded by $n$ and the size of the output of the separation oracle.

Designing strongly polynomial algorithms for continuous optimization problems with certain underlying combinatorial structure is a well-studied but challenging task in general. 
To this date, despite tremendous effort, it remains a major open question to solve linear programming (LP) in strongly polynomial time.
This problem is also widely known as Smale's 9th question~\cite{smale98}. 
Despite this barrier, such algorithms are known under additional assumptions: linear systems with at most two non-zero entries per row \cite{m83,ac91,cm94} or per column \cite{v17,ov20} in the constraint matrix, LPs with bounded entries in the constraint matrix \cite{t86,vy96,dhnv20}, and LPs with $0$-$1$ optimal solutions \cite{c12,c15}. 

For minimizing a general convex function $f$, strongly polynomial algorithms are hopeless unless $f$ satisfies certain combinatorial properties. 
In this work, we study the setting where the minimizer of $f$ is an integral point inside a box with radius\footnote{
It's easy to show that strongly polynomial algorithm doesn't exist if $\log(R)$ is super-polynomial (see Remark~\ref{remark:lower_bound}).} 
$R = 2^{\poly(n)}$. 
The integrality assumption on the minimizer is natural, and is general enough to encapsulate well-known problems such as submodular function minimization, where $R = 1$.
Prior to our work, an elegant application of simultaneous Diophantine approximation due to Gr\"otschel, Lov\'asz and Schrijver~\cite{gls84,gls88} gives\footnote{The original approach by Gr\"otschel, Lov\'asz and Schrijver was given in the context of obtaining exact solutions to LP, but it is immediately applicable to our problem. Their approach was briefly described in~\cite{gls84} with details given in~\cite{gls88}. Their approach originally used the ellipsoid method which is sub-optimal in terms of oracle complexity. The oracle complexity given here uses Vaidya's cutting plane method~\cite{v89}.} a strongly polynomial algorithm that minimizes $f$ using $O(n^2(n + \log(R)))$ calls to the separation oracle and an exponential time algorithm that finds the minimizer of $f$ using $O(n^2 \log(nR))$ oracle calls.

In fact, Gr\"otschel, Lov\'asz and Schrijver's approach applies to the more general setting of rational polyhedra, which they use to derive polynomial time algorithms for a wide range of combinatorial optimization problems \cite{gls81,gls88}. 
In the rational polyhedra setting, the set of minimizers of $f$ is a polyhedron $K^*$ inside a box with radius $R$, and the vertices of $K^*$ are all rational vectors with LCM vertex complexity\footnote{Here we use a slightly different definition from Gr\"otschel, Lov\'asz and Schrijver's original definition of vertex complexity in \cite{gls81,gls88} so that $\varphi = 0$ corresponds to the setting of integral minimizers. More details can be found in \Cref{subsubsec:rational_polyhedra}.} bounded by at most $\varphi \geq 0$ (\Cref{defn:vertex_complexity}). 
In particular, the case of integral minimizers in the previous paragraph corresponds to when $\varphi = 0$. 
For the more general setting of rational polyhedra, Gr\"otschel, Lov\'asz and Schrijver's approach implies a polynomial time algorithm that finds a vertex of $K^*$ using $O(n^2(n + \varphi + \log(R)))$ separation oracle calls, and an exponential time algorithm that uses $O(n^2(\varphi + \log(nR)))$ oracle calls. 
We refer interested readers to~\cite[Chapter 6]{gls88} for a detailed presentation of their approach.
The purpose of the present paper is to design a new method to improve the number of separation oracle calls.

A closely related problem, known as the Convex Integer Minimization problem, asks to minimize a convex function $f$ over the set of integer points.
Dadush~\cite[Section 7.5]{d12thesis} gave an algorithm for this problem that takes $n^{O(n)}$ time and exponential space. 
In fact, the Convex Integer Minimization problem generalizes integer linear programming and thus cannot be solved in sub-exponential time under standard complexity assumptions, so the integrality/rationality assumption on the minimizer of $f$ is, in some sense, necessary for obtaining efficient algorithms.

The number of separation oracle calls made by an algorithm for minimizing a convex function $f$, known as the {\em oracle complexity}, plays a central role in black-box models of convex optimization. 
For weakly polynomial algorithms, it's well-known that $\Theta(n \log(nR/\epsilon))$ oracle calls is optimal, with $\epsilon$ being the accuracy parameter.
The first exponential time algorithm that achieves the optimal oracle complexity is the famous center of gravity method discovered independently by Levin~\cite{l65} and Newman~\cite{n65}. 
As for polynomial time algorithms, an oracle complexity of this order was first achieved over thirty years ago by the method of inscribed ellipsoids~\cite{kte88,nn89}. 
In contrast, the optimal oracle complexity for strongly polynomial algorithms is largely unknown to this date.   
This motivates the present paper to place a focus on the oracle complexity aspect of our algorithms.

\subsection{Our results}
\label{subsec:result}

To formally state our result, we first define the notion of a separation oracle as formulated in \cite{gls81}. 

\begin{definition}[Separation oracle \cite{gls81}] \label{def:separation_oracle}
Let $f$ be a convex function on $\mathbb{R}^n$ and $K^*$ be the set of minimizers of $f$. 
Then a (strong) separation oracle $\SO$ for $f$ is one that:
\begin{itemize}
    \item [(a)] when queried with a minimizer $x \in K^*$, it outputs ``YES'';
    \item [(b)] when queried with a point $x \notin K^*$, it outputs a non-zero vector $c \in \mathbb{R}^n$ such that $\min_{y \in K^*} c^\top y > c^\top x$. 
\end{itemize}
\end{definition}


\medskip
\noindent\textbf{The setting of integral minimizers.} 
The main result of the paper in this setting is the following reduction to the Shortest Vector Problem (see Section~\ref{subsec:shortest_vector_lll}) given in \Cref{thm:main}. 
The seemingly strong assumption $(\star)$ guarantees that our algorithm finds an {\em integral} minimizer of $f$, which is crucial for our application to submodular function minimization.  
To find an arbitrary minimizer of $f$, we only need the much weaker assumption that $f$ has an integral minimizer (see \Cref{remark:any_minimizer}).

\begin{restatable}[Main result for integral minimizers]{theorem}{thmmain}  \label{thm:main}
Given a separation oracle $\SO$ for a convex function $f$ defined on $\mathbb{R}^n$, and a $\gamma$-approximation algorithm $\textsc{ApproxSVP}$ for the shortest vector problem which takes $T_{\textsc{SVP}}$ arithmetic operations.  
If the set of minimizers $K^*$ of $f$ is contained in a box of radius $R$ and satisfies
\begin{itemize}
    \item [($\star$)] all extreme points of $K^*$ are integral,
\end{itemize} 
then there is a randomized algorithm that with high probability finds an integral minimizer of $f$ using $O(n \log(\gamma n R))$ calls to $\SO$ and $\poly(n, \log(\gamma R)) \cdot T_{\textsc{SVP}}$ arithmetic operations.
\end{restatable}

In particular, taking $\textsc{ApproxSVP}$ to be the polynomial time $2^{n \log \log (n)/\log (n)}$-approximation algorithm in \cite{aks01} (which improves upon the celebrated 
LLL algorithm \cite{lll82} and Schnorr's block reduction algorithm \cite{s87}), or the exponential time algorithms for exact SVP \cite{aks01,mv13,adrs15} give the following corollary.

\begin{corollary}[Instantiations of main result] \label{cor:instantiation}
Under the same assumptions as in \Cref{thm:main}, there is a randomized algorithm that with high probability finds an integral minimizer of $f$ using
\begin{itemize}
\item[(a)] $O(n (n \log\log (n)/\log (n) + \log(R)) )$ calls to $\SO$ and $\poly(n, \log(R))$ arithmetic operations, or 
\item[(b)] $O(n \log(nR))$ calls to $\SO$ and $\exp(O(n)) \cdot \poly(\log(R))$ arithmetic operations. 
\end{itemize}
\end{corollary} 

More generally, for any integer $r > 1$, one can use the $r^{O(n/r)}$-approximation algorithm in $2^{O(r)} \poly(n)$ time for SVP given in \cite{aks01, mv13} to obtain a smooth tradeoff between time and oracle complexity in \Cref{thm:main}, but we omit the explicit statements of these results.

\begin{remark}[Assumption $(\star)$ and lower bound] \label{remark:lower_bound}
Without assumption $(\star)$, we give a $2^{\Omega(n)}$ information theoretic lower bound on the number of $\SO$ calls needed to find an integral minimizer of $f$. 
Consider the unit cube $K=[0,1]^n$ and let $V(K) = \{0,1\}^n$ be the set of vertices.  
For each $v \in V(K)$, define the simplex $\Delta(v) = \{x \in K: \norm{x - v}_1 < 0.01 \}$. 
Randomly pick a vertex $u \in V(K)$ and consider the convex function 
\begin{align*}
f_u(x) = \begin{cases}
0 \qquad & x \in K \setminus ( \cup_{v \in V(K) \setminus \{u\}} \Delta(v) )\\
\infty \qquad & \text{otherwise}
\end{cases}.
\end{align*}
When queried with a point $x \in \Delta(v)$ for some $v \in V(K) \setminus \{u\}$, we let $\SO$ output a separating hyperplane $H$ such that $K \cap H \subseteq \Delta(v)$; when queried with $x \notin K$, we let $\SO$ output a hyperplane that separates $x$ from $K$. 
Notice that $u$ is the unique integral minimizer of $f_u$, and to find $u$, one cannot do better than randomly checking vertices in $V(K)$ which takes $2^{\Omega(n)}$ queries to $\SO$. 

We next argue that $\Omega(n \log(R))$ calls to $\SO$ is information theoretically necessary in Theorem~\ref{thm:main}. 
Consider $f$ with a unique integral minimizer which is a random integral point in $B_\infty(R) \cap \mathbb{Z}^n$, where $B_\infty(R)$ is the $\ell_\infty$ ball with radius $R$. 
In this case, one cannot hope to do better than just bisecting the search space for each call to $\SO$ and this strategy takes $\Omega(n \log(R))$ calls to $\SO$ to reduce the size of the search space to a constant.  
\end{remark}

\begin{remark}[A weaker assumption] \label{remark:any_minimizer}
As shown in the previous remark, it is impossible in general to find an integral minimizer of $f$ efficiently without assumption $(\star)$. 
However, one can still find a  minimizer (which is not necessarily integral) of $f$under the much weaker assumption that $f$ has an integral minimizer, i.e. $K^* \cap \mathbb{Z}^n \neq \emptyset$. In this case, one can use the same algorithm as in \Cref{thm:main} until $\SO$ first returns ``YES'' and simply output the query point. The guarantees in \Cref{thm:main} also applies to this case. 
\end{remark}

\medskip
\noindent \textbf{Generalization to the rational polyhedra setting.} 
\Cref{thm:main} generalizes to the setting of rational polyhedra, where the set of minimizers $K^*$ of $f$ is a polyhedron contained in a box of radius $R$, and all vertices of $K^*$ are rational vectors with LCM vertex complexity at most $\varphi \geq 0$. 
Roughly speaking, this means that the least common multiple of the denominators in the fractional representation of each vertex is upper bounded by $2^\varphi$. 
We postpone the precise definitions of LCM vertex complexity and rational polyhedra to \Cref{subsubsec:rational_polyhedra} (\Cref{defn:vertex_complexity} and \ref{defn:rational_polyhedra}). 
The proof of the following theorem (which also implies \Cref{thm:main}) will be given in \Cref{sec:implementation}.

\begin{restatable}[Main result for rational polyhedra]{theorem}{thmpolyhedra}  \label{thm:polyhedra}
Given a separation oracle $\SO$ for a convex function $f$ defined on $\mathbb{R}^n$, and a $\gamma$-approximation algorithm $\textsc{ApproxSVP}$ for the shortest vector problem which takes $T_{\textsc{SVP}}$ arithmetic operations.  
If the set of minimizers $K^*$ of $f$ is a rational polyhedron contained in a box of radius $R$ and has LCM vertex complexity at most $\varphi \geq 0$, then there is a randomized algorithm that with high probability finds a vertex of $K^*$ using $O(n(\varphi + \log(\gamma n R)))$ calls to $\SO$ and $\poly(n, \varphi, \log(\gamma R)) \cdot T_{\textsc{SVP}}$ arithmetic operations.
\end{restatable}

\subsection{Application to Submodular Function Minimization}
\label{subsec:appliaction_submodular}

Submodular function minimization (SFM) has been recognized as an important problem in the field of combinatorial optimization. 
Classical examples of submodular functions include graph cut functions, set coverage function, and utility functions from economics. 
Since the seminal work by Edmonds in 1970~\cite{e70}, SFM has served as a popular tool in various fields such as theoretical computer science, operations research, game theory, and machine learning. 
For a more comprehensive account of the rich history of SFM, we refer interested readers to the excellent surveys~\cite{m05,i08}.

\begin{table}[htp!]
    \centering
    \begin{tabular}{ | l | l | l | l | l | }
        \hline
        {\bf Authors} & {\bf Year} & {\bf Oracle Complexity} & {\bf Remarks} \\ \hline \hline
        Gr\"otschel, Lov\'asz, Schrijver~\cite{gls81,gls88}    & 1981,88    & $\widetilde{O}(n^5)$~\cite{m05} & first strongly    \\ \hline
        Schrijver~\cite{s00} & 2000  & $O(n^8)$ & first comb. strongly \\ \hline
        Iwata, Fleischer, Fujishige~\cite{iff01}         & 2000    & $ O(n^7 \log(n))$  & first comb. strongly \\ \hline
        Fleischer, Iwata~\cite{fi03}  & 2000    & $O(n^7)$   &  \\ \hline
        Iwata~\cite{i03} & 2002 & $O(n^6 \log(n))$  &   \\ \hline
        Vygen~\cite{v03}   & 2003    & $ O(n^7)$ &  \\ \hline
        Orlin~\cite{o09}   & 2007    & $O(n^5)$  &  \\ \hline
        Iwata, Orlin~\cite{io09} & 2009 & $ O(n^5 \log(n)) $ & \\ \hline
        Lee, Sidford, Wong~\cite{lsw15} & 2015 & $O(n^3 \log^2(n))$ & current best strongly \\ \hline
        Lee, Sidford, Wong~\cite{lsw15} & 2015 & $O(n^3 \log(n))$ & exponential time \\ \hline
        Dadush, V{\'e}gh, Zambelli~\cite{dvz18} & 2018 & $O(n^3 \log^2(n))$ & current best strongly \\ \hline
        {\bf This paper} & 2020 & $O(n^3 \log\log (n)/\log (n))$ &  \\ \hline 
        {\bf This paper} & 2020 & $O(n^2 \log(n))$ & exponential time \\ \hline
    \end{tabular}
    \caption{Strongly polynomial algorithms for submodular function minimization. The oracle complexity measures the number of calls to the evaluation oracle $\EO$. In the case where a paper is published in both conference and journal, the year we provide is the earliest one.}
    \label{tab:submodular}
\end{table}

The formulation of SFM we consider is the standard one: we are given a submodular function $f$ defined over subsets of an $n$-element ground set. 
The values of $f$ are integers, and are evaluated by querying an evaluation oracle that takes time $\EO$. 
Since the breakthrough work by Gr\"otschel, Lov\'asz, Schrijver~\cite{gls81,gls88} that the ellipsoid method can be used to construct a strongly polynomial algorithm for SFM, there has been a vast literature on obtaining better strongly polynomial algorithms (see Table~\ref{tab:submodular}). 
These include the very first combinatorial strongly polynomial algorithms constructed by Iwata, Fleischer and Fujishige~\cite{iff01} and Schrijver~\cite{s00}.
Very recently, a major improvement was made by Lee, Sidford and Wong~\cite{lsw15} using an improved cutting plane method. 
Their algorithm achieves the state-of-the-art oracle complexity of $O(n^3 \log^2(n))$ for strongly polynomial algorithms. 
A simplified variant of this algorithm achieving the same oracle complexity was given in~\cite{dvz18}.

The authors of~\cite{lsw15} also noted that $O(n^3 \log(n))$ oracle calls are information theoretically sufficient for SFM (\cite[Theorem 71]{lsw15}), but were unable to give an efficient algorithm achieving such an oracle complexity. 
They asked as open problems (\cite[Section 16.1]{lsw15}):
\begin{itemize}
\item [(a)] whether there is a strongly polynomial algorithm achieving the $O(n^3 \log(n))$ oracle complexity; 
\item [(b)] whether one could further (even information theoretically) remove the extraneous $\log(n)$ factor from the oracle complexity. 
\end{itemize}
The significance of these questions stem from their belief that $\Theta(n^3)$ is the tight oracle complexity for strongly polynomial algorithms for SFM (see \cite[Section 16.1]{lsw15} for a more detailed discussion).

We answer both these open questions affirmatively 
in the following Theorem~\ref{thm:submodular_main}, which follows from applying \Cref{cor:instantiation} to the Lov\'asz extension $\hat{f}$ of the function $f$, together with the standard fact that a separation oracle for $\hat{f}$ can be implemented using $n$ calls to the evaluation oracle (\cite[Theorem 61]{lsw15}).
We provide details on these definitions and the proof of \Cref{thm:submodular_main} in \Cref{sec:submodular_appendix}.

\begin{restatable}[Submodular function minimization]{theorem}{Submodular} \label{thm:submodular_main}
Given an evaluation oracle $\EO$ for a submodular function $f$ defined over subsets of an $n$-element ground set, there exist 
\begin{itemize}
\item[(a)] a strongly polynomial algorithm that minimizes $f$ using $O(n^3 \log\log (n)/\log (n))$ calls to $\EO$, and 
\item[(b)] an exponential time algorithm that minimizes $f$ using $O(n^2 \log(n))$ calls to $\EO$. 
\end{itemize}
\end{restatable}

To the best of our knowledge, the results in \Cref{thm:submodular_main} represent the first algorithms that achieve $o(n^3)$ oracle complexity for SFM, even information theoretically. 
The first result in \Cref{thm:submodular_main} breaks the natural $O(n^3)$ barrier for the oracle complexity of strongly polynomial algorithms.  
The second result pushes the information theoretic oracle complexity for exact SFM down to nearly quadratic.

Our algorithm is conceptually simpler than the algorithms given in~\cite{lsw15,dvz18}. 
Moreover, while most of the previous strongly polynomial algorithms for SFM vastly exploit different combinatorial structures of submodularity, our result is achieved via a very general algorithm and uses the structural properties of submodular functions in a minimal way.

\subsection{Proof Overview}
\label{subsec:proof_overview}

Without loss of generality, we may assume that $f$ has a unique minimizer $x^*$ in \Cref{thm:main} and \ref{thm:polyhedra}. 
To justify this statement, suppose the set of minimizers $K^*$ of $f$ satisfies assumption $(\star)$.
Let $x^* \in K^*$ be the unique lexicographically minimal minimizer, i.e. every other minimizer $x \in K^*$ satisfies $x_i > x_i^*$ for the smallest coordinate $i \in [n]$ in which $x_i \neq x_i^*$. 
Whenever $\SO$ is queried at a minimizer $y \in K^*$ and outputs ``YES'', our algorithm continues to minimize the linear objective $e_i^\top x$, where $i \in [n]$ is the smallest index such that the $i$th standard orthonormal basis vector $e_i$ is not orthogonal to the current working subspace, by pretending that $\SO$ returns\footnote{Note that this implementation of the separation oracle for the lexicographically minimal minimizer $x^*$ does not quite satisfy the conditions in \Cref{def:separation_oracle}. In particular, even when $x^*$ is queried, the separation oracle for finding $x^*$ might not realize it unless the current working subspace is trivial (i.e. $0$-dimensional). However, all our results and proofs still hold under this slightly weaker implementation of the separation oracle.} the vector $-e_i$ (until its search set contains a single point). 
Equivalently, our algorithm minimizes the linear objectives $e_1^\top x, \cdots, e_n^\top x$ in the given order inside $K^*$, and this optimization problem has the unique solution $x^*$. 
We make the assumption that $f$ has a unique minimizer $x^*$ in the rest of this paper. 


For simplicity, we further assume in the subsequent discussions that $x^* \in \{0,1\}^n$, i.e. $R = 1$ in the setting of integral minimizer, which does not change the problem inherently.

On a high level, our algorithm maintains a convex search set $K$ that contains the integral minimizer $x^*$ of $f$, and iteratively shrinks $K$ using the cutting plane method; as the volume of $K$ becomes small enough, our algorithm finds a hyperplane $P$ that contains all the integral points in $K$ and recurse on the lower-dimensional search set $K \cap P$.
The assumption that $x^*$ is integral guarantees that $x^* \in K \cap P$.
This natural idea was previously used in~\cite{gls84,gls88} to handle rational polytopes that are not full-dimensional and in~\cite{lsw15} to argue that $O(n^3 \log(n))$ oracle calls is information theoretically sufficient for SFM. 
The main technical difficulties in efficiently implementing such an idea are two-fold: 
\begin{itemize}
    \item [(a)] we need to find the hyperplane $P$ that contains $K \cap \mathbb{Z}^n$;
    \item [(b)] we need to carefully control the amount $\vol(K)$ is shrunk so that progress is not lost.   
\end{itemize}
The second difficulty is key to achieving a small oracle complexity and deserves some further explanation. 
To see why shrinking $K$ arbitrarily might result in a loss of progress, it's instructive to consider the following toy example: suppose an algorithm starts with the unit cube $K = [0,1]^n$ and $x^*$ lies on the hyperplane $K_1 = \{x: x_1 = 0\}$; 
suppose the algorithm obtains, in its $i$th call to $\SO$, the halfspace $H_i = \{x: x_1 \leq 2^{-i}\}$.  
After $T$ calls to $\SO$, the algorithm obtains the refined search set $K \cap H_T$ with volume $2^{-T}$. 
However, when the algorithm reduces the dimension and recurses on the hyperplane $K_1$, the $(n-1)$-dimensional volume of the search set again becomes $1$, and
the progress made by the algorithm in shrinking the volume of $K$ is entirely lost. 
In contrast, the correct algorithm can reduce the dimension after only one call to $\SO$ when it's already clear that $x^* \in K_1$.

\subsubsection{The Gr\"otschel-Lov\'asz-Schrijver Approach}

For the moment, let's take $K$ to be an ellipsoid. 
Such an ellipsoid can be obtained by Vaidya's volumetric center cutting plane method\footnote{Perhaps a more natural candidate is the ellipsoid method developed in~\cite{yn76,s77,k80}. This method, however, shrinks the volume of $K$ by a factor of $O(n)$ slower than Vaidya's method. In fact, the Gr\"otschel-Lov\'asz-Schrijver approach~\cite{gls84} originally used the ellipsoid method which results in an oracle complexity of $O(n^4)$ for their polynomial time algorithm.}~\cite{v89}. 
One natural idea to find the hyperplane comes from the following geometric intuition: when the ellipsoid 
$K$ is ``flat'' enough in one direction, then all of its integral points lie on a hyperplane $P$. 
To find such a hyperplane $P$, Gr\"otschel, Lov\'asz and Schrijver~\cite{gls84,gls88} gave an elegant application of simultaneous Diophantine approximation.  
We explain the main ideas behind this application in the following. 
We refer interested readers to~\cite[Chapter 6]{gls88} for a more comprehensive presentation of their approach and its implications to finding exact LP solutions. 

For simplicity, we assume $K$ is centered at $0$.  
Let $a$ be the unit vector parallel to the shortest axis of $K$ and $\mu_{\min}$ be the Euclidean length of the shortest axis of $K$. 
Approximating the vector $a$ using the efficient simultaneous Diophantine approximation algorithm by Lenstra, Lenstra and Lov\'asz~\cite{lll82}, 
one obtains an integral vector $v \in \mathbb{Z}^n$ and a positive integer $q \in \mathbb{Z}$ such that
\begin{align*}
\norm{q a - v}_\infty < 1/3n \qquad \text{and} \qquad  0 < q < 2^{2n^2} .
\end{align*}
This implies that for any integral point $x \in K \cap \{0,1\}^n$, 
\begin{align*}
|v^\top x| \leq |qa^\top x| + \frac{1}{3n} \cdot \norm{x}_1 \leq q \cdot \mu_{\min} + 1/3 .
\end{align*}

When $\mu_{\min} < 2^{-3n^2}$, the integral inner product $v^\top x$ has to be $0$ and therefore all integral points in $K$ lie on the hyperplane $P = \{x: v^\top x = 0\}$. 
An efficient algorithm immediately follows: we first run the cutting plane method until the shortest axis of $K$ has length $\mu_{\min} \approx 2^{-3n^2}$, then apply the above procedure to find the hyperplane $P$ on which we recurse. 

To analyze the oracle complexity of this algorithm, one naturally uses $\vol(K)$ as the potential function. 
An amortized analysis using such a volume potential previously appeared, for example, in~\cite{dvz20} for finding maximum support solutions in the linear conic feasibility problem. 
Roughly speaking, each cutting plane step (corresponding to one oracle call) decreases $\vol(K)$ by a constant factor; each dimension reduction step increases $\vol(K)$ by roughly $1/\mu_{\min} \approx 2^{3n^2}$. 
As there are $n$ dimension reduction steps before the problem becomes trivial, the total number of oracle calls is thus $O(n^3)$.  
The exponential time oracle complexity bound of $O(n^2 \log (n))$ can be obtained similarly by using Dirichlet's approximation theorem on simultaneous Diophantine approximation (e.g. \cite[Section 1.10]{c71}) instead.

One might wonder if the oracle complexity upper bound for their polynomial time algorithm can be improved using a better analysis.   
However, there is some fundamental issue in getting such an improvement. 
In particular, the upper bound of $2^{O(n^2)}$ on $q$ in efficient simultaneous Diophantine approximation corresponds to the $2^{O(n)}$-approximation factor of the Shortest Vector Problem in lattices, first obtained by Lenstra, Lenstra and Lov\'asz~\cite{lll82}. 
Despite forty years of effort, this approximation factor has only been improved slightly to $2^{n \log \log (n)/\log n}$ for polynomial time algorithms \cite{aks01}.

\subsubsection{Lattices to the Rescue: A Reduction to the Shortest Vector Problem}

\label{subsubsec:basic_algo_tech}
To bypass the previous bottleneck and prove Theorem~\ref{thm:main}, we give a reduction to the Shortest Vector Problem directly. 
We give a new method to find the hyperplane for dimension reduction based on an approximately shortest vector of certain lattice, and analyze its oracle complexity via a novel potential function that captures simultaneously the volume of the search set $K$ and the density of the lattice. 
The change in the potential function after dimension reduction is analyzed through a high dimensional slicing lemma. 
The details for this algorithm and its analysis are given in \Cref{sec:meta_algo} and \ref{sec:implementation}.

\medskip

\noindent \textbf{Finding the hyperplane.} We maintain a polytope $K$ (which we assume to be centered at $0$ for simplicity) using an efficient implementation of the center of gravity method due to Bertsimas and Vempala~\cite{bv04}. 
The following sandwiching condition is standard in convex geometry
\begin{align} \label{eq:sandwiching_intro}
E(\Cov(K)^{-1}) \subseteq K \subseteq 2n \cdot E(\Cov(K)^{-1}),
\end{align}
where $\Cov(K)$ is the covariance matrix of the uniform distribution over $K$. 
Sufficiently good approximation to $\Cov(K)$ can be obtained efficiently by sampling from $K$~\cite{bv04} so we ignore any computational issue for now.

To find a hyperplane $P$ that contains all integral points in $K$, it suffices to preserve all the integral points in the outer ellipsoid $E = 2n \cdot E(\Cov(K)^{-1})$ on the RHS of \eqref{eq:sandwiching_intro}. 
Let $x \in E \cap \mathbb{Z}^n$ be an arbitrary integral point. 
For any vector $v$, 
\begin{align} \label{eq:ellipsoid_width}
|v^\top x| \leq \norm{v}_{\Cov(K)} \cdot \norm{x}_{\Cov(K)^{-1}} \leq 2n \cdot \norm{v}_{\Cov(K)}.
\end{align}
As long as $\norm{v}_{\Cov(K)} < 1/10n$ and $v^\top x$ is an integer, we can conclude that $v^\top x = 0$ and this implies that all integral points in $K$ lie on the hyperplane $P = \{x: v^\top x = 0\}$. 
Note that by \eqref{eq:ellipsoid_width}, such a vector $v$ with small $\|v\|_{\Cov(K)}$ essentially controls the ellipsoid width $\width_E(v) := \max_{x \in E} v^\top x - \min_{x \in E} v^\top x$. 

One might attempt to guarantee that $v^\top x$ is integral by choosing $v$ to be an integral vector. 
However, this idea has a fundamental flaw: as the algorithm reduces the dimension by restricting on a subspace $W$, the set of integral points on $W$ might become much {\em sparser}. 
As such, one needs $\vol(K)$ to be very small to guarantee that $\norm{v}_{\Cov(K)} < 1/10n$ and this results in a very large oracle complexity.

To avoid this issue, we take $v = \Pi_W(z) \neq 0$ as the projection of some integral point $z \in \mathbb{Z}^n$ on $W$, where $W$ is the subspace on which $K$ lies. 
Since $z - v \in W^\bot$, we have $v^\top x = z^\top x$ and this guarantees that $v^\top x$ is integral. 
For the general case where $K$ is not centered at $0$, a simple rounding procedure computes the desired hyperplane. 
We postpone the details of constructing the hyperplane to \Cref{lem:reduce_dimension}.

How do we find a vector $v \in \Pi_W(\mathbb{Z}^n) \setminus \{0\}$ that satisfies $\norm{v}_{\Cov(K)} < 1/10n$? 
This is where lattices come into play. 
In particular, since $ \Lambda = \Pi_W(\mathbb{Z}^n)$ forms a lattice, we can apply any $\gamma$-approximation algorithm for the Shortest Vector Problem. 
If the shortest non-zero vector in $\Lambda$ has $\Cov(K)$-norm at most $1/10 \gamma n$, then we can find a non-zero vector $v$ that satisfies $\norm{v}_{\Cov(K)} < 1/10n$. 

\medskip 
\noindent \textbf{The algorithm.}
This new approach for finding the hyperplane immediately leads to the following algorithm: 
we run the approximate center of gravity method for one step to decrease the volume of the polytope $K$ by a constant factor; then we run the $\gamma$-approximation algorithm for SVP to find a non-zero vector $v$ for dimension reduction.
If $\norm{v}_{\Cov(K)} \geq 1/10n$, then we continue to run the cutting plane method; otherwise, we use the above procedure to find a hyperplane $P$ containing all integral points in $K$, update the polytope $K$ to be $K \cap P$ and recurse.

\medskip 
\noindent \textbf{Potential function analysis.}  
To analyze such an algorithm, one might attempt to use $\vol(K)$ as the potential function as in the Gr\"otschel-Lov\'asz-Schrijver approach. 
However, one quickly realizes that $\vol(K \cap P) / \vol(P)$ can be as large as $\norm{v}_2 / \norm{v}_{\Cov(K)}$. 
While it's expectable that $\norm{v}_{\Cov(K)}$ is not too small since we are frequently checking for a short lattice vector, one has no control over $\norm{v}_2$ in general. 

Key to our analysis is the potential function $\Phi = \vol(K) \cdot \det(\Lambda)$ that measures simultaneously the volume of $K$ and the covolume $\det(\Lambda)$ of the lattice $\Lambda$. 
Essentially, this potential function controls the lattice width $\min_{v \in \Lambda \setminus \{0\}} \width_E(v)$ of the outer ellipsoid $E$. 
In fact, Minkowski's first theorem (\Cref{thm:minkowski_shortest_vector}) implies that there always exists a vector $v \in \Lambda \setminus \{0\}$ such that $\width_E(v) \leq \poly(n) \cdot \Phi^{1/n}$, and thus the potential function would never get too small before dimension reduction takes place.  

Continuing with the analysis via the potential function $\Phi$, while $\vol(K)$ increases by $\norm{v}_2 / \norm{v}_{\Cov(K)}$ after the dimension reduction, standard fact on lattice projection (Fact~\ref{fact:lattice_projection}) shows that the covolume of the lattice decreases by a factor of $\norm{v}_2$.
The decrease in the covolume of the lattice thus elegantly cancels out the increase in $\vol(K)$, leading to an overall increase in the potential of at most $1/\norm{v}_{\Cov(K)} = O(\gamma n)$.  
It follows that the total increase in the potential over all $n$ dimension reduction steps is at most $(\gamma n)^n$.
Note that each cutting plane step still decreases the potential function by a constant factor since the lattice is unchanged. 
Therefore, the total number of oracle calls is at most $O(n \log(\gamma n))$. 

\medskip 
\noindent \textbf{High dimensional slicing lemma for consecutive dimension reduction steps.} 
The argument above ignores a slight technical issue: while we can guarantee that $\norm{v}_{\Cov(K)} \geq 1/\gamma n$ after cutting plane steps by checking for short non-zero lattice vectors, it's not clear why $\norm{v}_{\Cov(K)}$ cannot be too small after a sequence of dimension reduction steps.  
It turns out that this can happen only when $\Cov(K)$ becomes much smaller (e.g. the hyperplane $P$ is far from the centroid of $K$) after dimension reduction, in which case $\vol(K)$ as well as the potential also become much smaller. 

To formally analyze the change in the potential function after a sequence of $k$ consecutive dimension reduction steps, we note that the polytope $K$ (which we assume to be isotropic for simplicity) becomes a ``slice'' $K \cap W$ and the lattice $\Lambda$ becomes the projected lattice $\Pi_W(\Lambda)$, where $W$ is a subspace. 
One can show using standard convex geometry tools that $\vol(K \cap W) / \vol(K)$ is at most $k^{O(k)}$, and via Minkowski's first theorem that $\det(\Pi_W(\Lambda)) / \det(\Lambda)$ is at most $\sqrt{k}^k / \lambda_1(\Lambda)^k$, where $\lambda_1(\Lambda)$ is the Euclidean length of the shortest non-zero vector in $\Lambda$. 
We leave the details of this high dimensional slicing lemma to Lemma~\ref{lem:high_dim_slicing}.  
Since we know that $\lambda_1(\Lambda) \geq 1/\gamma n$ in the first dimension reduction step, the potential function increases by a factor of at most $(\gamma n)^{O(k)}$ over a sequence of $k$ consecutive dimension reduction steps. 
This gives a more precise analysis of the $O(n \log(\gamma n))$ oracle complexity.


\section{Preliminaries}

\subsection{Notations}

We use $\mathbb{R}_+$ to denote the set of non-negative real numbers. 
For any positive integer $n$, we use $[n]$ to denote the set $\{1, \cdots, n\}$.
Given a real number $a \in \mathbb{R}$, the floor of $a$, denoted as $\lfloor a \rfloor$, is the largest integer that is at most $a$. 
Define the closest integer to $a$, denoted as $\lceil a \rfloor$, to be $\lceil a \rfloor := \lfloor a + 1/2 \rfloor$. 
Given an integer $\varphi \geq 0$ and $a \in \mathbb{R}$, we use $\lceil a \rfloor_\varphi$ to denote the closest rational number to $a$ with denominator at most $2^\varphi$.  
Given integers $a_1, \cdots, a_m$ which are not all $0$, we denote $\gcd(a_1, \cdots, a_m)$ their greatest common divisor.
Given non-zero integers $a_1, \cdots, a_m$, we denote $\lcm(a_1, \cdots, a_m)$ their least common multiple. 

For any $i \in [n]$, we denote $e_i$ the $i$th standard orthonormal basis vector of $\mathbb{R}^n$.  
We use $B_p(R)$ to denote the $\ell_p$-ball of radius $R$ in $\mathbb{R}^n$ and $B_p = B_p(1)$ the unit $\ell_p$-ball. 
For any set of vectors $V \subseteq \mathbb{R}^n$, we use $\spn \{V\}$ to denote the linear span of vectors in $V$. 
Throughout, a subspace $W$ is a linear subspace of $\mathbb{R}^n$ with $0 \in W$; an affine subspace $W$ is a translation of a subspace of $\mathbb{R}^n$ (and thus might not pass through the origin). 
Given a subspace $W$, we denote $W^\bot$ the orthogonal complement of $W$ and $\Pi_W(\cdot)$ the orthogonal projection onto the subspace $W$.
Given a PSD matrix $A \in \mathbb{R}^{n \times n}$ and a subspace $V \subseteq \mathbb{R}^n$, we say $A$ has full rank on $V$ if $\rank(A)=\dim(V)$ and the eigenvectors corresponding to non-zero eigenvalues of $A$ form an orthogonal basis of $V$.  

Given a subspace $V \subseteq \mathbb{R}^n$ and a PSD matrix $A \in \mathbb{R}^{n \times n}$ that has full rank on $V$, 
the function $\langle \cdot, \cdot \rangle_A$ given by $\langle x, y \rangle_A = x^\top A y$ defines an inner product on $V$. 
The inner product $\langle \cdot, \cdot \rangle_A$ induces a norm on $V$, i.e. $\norm{x}_A = \sqrt{\langle x, x \rangle_A}$ for any $x \in V$, which we call the $A$-norm. 
Given a point $x_0 \in \mathbb{R}^n$ and a PSD matrix $A \in \mathbb{R}^{n \times n}$, we use $E(x_0, A)$ to denote the (might not be full-rank) ellipsoid given by $E(x_0, A) := \{x \in x_0 + W_A: (x - x_0)^\top A (x - x_0) \leq 1\}$, where $W_A$ is the subspace spanned by eigenvectors corresponding to non-zero eigenvalues of $A$. 
When the ellipsoid is centered at $0$, we use the short-hand notation $E(A)$ to denote $E(0,A)$.

\subsection{Lattices}
Given a set of linearly independent vectors $b_1, \cdots, b_k \in \mathbb{R}^n$, denote $\Lambda(b_1, \cdots, b_k) = \{\sum_{i=1}^k \lambda_i b_i, \lambda_i \in \mathbb{Z}\}$ the lattice generated by $b_1, \cdots, b_k$. 
Here, $k$ is called the rank of the lattice. 
A lattice is said to have full-rank if $k = n$.
Any set of $k$ linearly independent vectors that generates the lattice $\Lambda = \Lambda(b_1, \cdots, b_k)$ under integer linear combinations is called a basis of $\Lambda$.
In particular, the set $\{b_1, \cdots, b_k\}$ is a basis of $\Lambda$. 
Different basis of a full-rank lattice are related by unimodular matrices, which are integer matrices with determinant $\pm 1$.

Given a basis $B \in \mathbb{R}^{n \times k}$, the fundamental parallelepiped of $\Lambda = \Lambda(B)$ is the polytope $\mathcal{P}(B):=\{\sum_{i=1}^k \lambda_i b_i: \lambda_i \in [0,1), \forall i \in [k]\}$. 
The determinant of the lattice (also known as the covolume), denoted as $\det(\Lambda)$, is defined to be the volume of the fundamental parallelepiped, which is independent of the basis.
We also define the notion of dual lattices below.

\begin{definition}[Dual lattice]
Given a lattice $\Lambda \subseteq \mathbb{R}^n$, the dual lattice $\Lambda^*$ is the set of all vectors $x \in \spn\{\Lambda\}$ such that $\langle x, y \rangle \in \mathbb{Z}$ for all $y \in \Lambda$. 
\end{definition}

We refer interested readers to standard textbooks (e.g. \cite{s98}) for a more comprehensive introduction to lattice theory.



\subsubsection{Lattice Projection and Intersection with Subspaces} 

The following standard facts on lattice projection follow from Gram-Schmidt orthogonalization.

\begin{fact}[Lattice projection] \label{fact:lattice_projection}
Let $\Lambda$ be a full-rank lattice in $\mathbb{R}^n$ and $W$ be a linear subspace such that $\dim(\spn\{\Lambda \cap W \}) = \dim(W)$. 
Then we have 
\begin{align*}
\det(\Lambda) = \det(\Lambda \cap W) \cdot \det(\Pi_{W^\bot}(\Lambda)) .
\end{align*}
\end{fact}

\begin{fact}[Dual of lattice projection] \label{fact:lattice_duality}
Let $\Lambda$ be a full-rank lattice in $\mathbb{R}^n$ and $W$ be a linear subspace such that $\dim(\spn\{\Lambda \cap W \}) = \dim(W)$. 
Then we have the following duality 
\begin{align*}
(\Pi_W (\Lambda))^* = \Lambda^* \cap W.
\end{align*}
\end{fact}

\subsubsection{Minkowski's First Theorem}

Minkowski's first theorem~\cite{m53} asserts the existence of a non-zero lattice point in a symmetric convex set with large enough volume. 
An important consequence of it is the following upper bound on $\lambda_1(\Lambda, A)$, the length of the shortest non-zero vector in lattice $\Lambda$ under $A$-norm.

\begin{theorem}[Consequence of Minkowski's first theorem,~\cite{m53}] \label{thm:minkowski_shortest_vector} 
Let $\Lambda$ be a full-rank lattice in $\mathbb{R}^n$ and $A \in \mathbb{R}^{n \times n}$ be a positive definite matrix. 
Then 
\begin{align*}
\lambda_1(\Lambda, A) \leq \sqrt{n} \cdot \det( A^{1/2})^{1/n} \cdot \det(\Lambda)^{1/n} .
\end{align*}
\end{theorem}

\subsubsection{The Shortest Vector Problem and the Lenstra-Lenstra-Lov\'asz Algorithm}
\label{subsec:shortest_vector_lll}

Given a lattice $\Lambda$ and a PSD matrix $A$ that has full rank on $\spn\{\Lambda\}$, the Shortest Vector Problem (SVP) asks to find a shortest non-zero vector in $\Lambda$ under $A$-norm\footnote{Equivalently, one could think of finding an approximately shortest vector under the Euclidean norm in the lattice $A^{1/2} \Lambda$.}, whose length is denoted as $\lambda_1(\Lambda, A)$. 
SVP is one of the most fundamental computational problems in lattice theory and is known to be NP-hard.  
For this problem, the celebrated Lenstra-Lenstra-Lov\'asz (LLL) algorithm~\cite{lll82} finds in polynomial time a $2^{n/2}$-approximation to $\lambda_1(\Lambda,A)$. 
Building on top of a block-reduction algorithm by Schnorr \cite{s87}, Ajtai, Kumar and Sivakumar \cite{aks01} obtained the current best  polynomial time approximation factor of $2^{n \log \log (n)/\log (n)}$ for SVP.

\begin{theorem}[\cite{aks01}] \label{thm:lll}
Given a basis $b_1, \cdots, b_n \in \mathbb{Z}^n$ for lattice $\Lambda$ and a positive definite matrix $A \in \mathbb{Z}^{n \times n}$. 
Let $D \in \mathbb{Z}$ be such that $\norm{b_i}_A^2 \leq D$ for any $i \in [n]$. 
Then there exists an algorithm that outputs in $\poly(n, \log(D))$ arithmetic operations a vector $b_1'$ such that 
\begin{align*}
\norm{b_1'}_A \leq 2^{n \log\log (n)/\log (n)} \cdot \lambda_1(\Lambda,A) .
\end{align*} 
Moreover, the integers occuring in the algorithm have bit sizes at most $\poly(n, \log(D))$.
\end{theorem}

In fact, for any integer $r > 1$, \cite{aks01} gave a $2^{O(r)}\poly(n)$-time $r^{O(n/r)}$-approximation algorithm for SVP, allowing a smooth tradeoff between time and approximation quality. 

For solving SVP exactly, the state-of-the-art is a deterministic $\widetilde{O}(2^{2n})$-time and $\widetilde{O}(2^n)$-space algorithm given by Micciancio and Voulgaris~\cite{mv13}, and a randomized $2^{n+o(n)}$-time and space algorithm due to Aggarwal et al.~\cite{adrs15}.
We refer to these excellent papers and the references therein for a comprehensive account of the rich history of SVP.

\subsubsection{Rational Polyhedra} \label{subsubsec:rational_polyhedra}

We start with the definition of the LCM vertex complexity of a rational vector. 

\begin{definition}[LCM vertex complexity] \label{defn:vertex_complexity}
Given a rational vector $a = (p_1/q_1, \cdots, p_n/q_n)$, where integers $p_i$ and $q_i \geq 1$ are coprime for all $i \in [n]$, we define its LCM vertex complexity to be the smallest integer $\varphi \geq 0$ such that the 1-dimensional lattice $L_a := \{a^\top z: z \in \mathbb{Z}^n\}$ is a sub-lattice of $\mathbb{Z}/q$ for some positive integer $q \leq 2^{\varphi}$. 
\end{definition}

In particular, the number $q$ above is $\lcm(q_1, \cdots, q_n)$. When $\gcd(p_1, \cdots, p_n) = 1$, by B\'ezout's identity, we in fact have that $L_a = \mathbb{Z}/q$. 
We next formally define the notion of rational polyhedra with bounded LCM vertex complexity. 

\begin{definition}[Rational polyhedra with bounded LCM vertex complexity] \label{defn:rational_polyhedra}
A bounded convex set $K \subseteq \mathbb{R}^n$ is a rational polyhedron with LCM vertex complexity at most $\varphi \geq 0$ if $K$ is a polyhedron and the LCM vertex complexity of every vertex of $K$ is at most $\varphi$. 
\end{definition}

For convenience, we define the set of all rational vectors with bounded LCM vertex complexity. 

\begin{definition}[Rational vectors with bounded LCM vertex complexity] \label{defn:Sphi}
For any integer $\varphi \geq 0$, we define $S_\varphi^n$ the set of all rational vectors in $\mathbb{R}^n$ with LCM vertex complexity at most $\varphi$. 
\end{definition}

\begin{remark} [Different definitions] \label{remark:compare_def}
We remark that our definition of LCM vertex complexity in \Cref{defn:vertex_complexity} is different from the standard definition of vertex complexity in the literature used by Gr\"otschel, Lov\'asz and Schrijver \cite{gls88}, who defined the vertex complexity of a rational vector $a$ to be its binary description length, i.e. bit complexity. 
The LCM vertex complexity of a rational vector as in \Cref{defn:vertex_complexity} is always smaller than its bit complexity, and in fact might be much smaller. 
The reason we deviate from Gr\"otschel, Lov\'asz and Schrijver's more standard notion of vertex complexity is that \Cref{defn:vertex_complexity} allows a slightly cleaner presentation of the results and proofs in this paper. In particular, one can obtain the results and proofs in the setting of integral minimizers by taking $\varphi = 0$. 
\end{remark}


\subsection{Convex Geometry}

A function $g: \mathbb{R}^n \rightarrow \mathbb{R}_+$ is log-concave if its support $\supp(g)$ is convex and $\log(g)$ is concave on $\supp(g)$. 
An integrable function $g: \mathbb{R}^n \rightarrow \mathbb{R}_+$ is a density function, if $\int_{\mathbb{R}^n} g(x) dx = 1$. 
The centroid of a density function $g: \mathbb{R}^n \rightarrow \mathbb{R}_+$ is defined as $\cg(g) = \int_{\mathbb{R}^n} g(x) x dx$; the covariance matrix of the density function $g$ is defined as $\Cov(g) = \int_{\mathbb{R}^n} g(x) (x - \cg(g)) (x - \cg(g))^\top d x$. 
A density function $g: \mathbb{R}^n \rightarrow \mathbb{R}_+$ is isotropic, if its centroid is $0$ and its covariance matrix is the identity matrix, i.e. $\cg(g) = 0$ and $\Cov(g) = I$.

A typical example of a log-concave distribution is the uniform distribution over a convex body $K \subseteq \mathbb{R}^n$.  
Given a convex body $K$ in $\mathbb{R}^n$, its volume is denoted as $\vol(K)$. 
The centroid (resp. covariance matrix) of $K$, denoted as $\cg(K)$ (resp. $\Cov(K)$), is defined to be the centroid (resp. covariance matrix) of the uniform distribution over $K$. 
A convex body $K$ is said to be isotropic if the uniform density over it is isotropic. 
Any convex body can be put into its isotropic position via an affine transformation.

Sometimes we will be working with a bounded convex set $K \subseteq W$, where $W$ is an affine subspace that might not be full dimensional. 
For convenience, we extend the definitions above to this case by first applying a linear transformation and then restricting to $W$ so that $K$ becomes full-dimensional.

\begin{theorem}[Brunn's principle] \label{thm:brunn}
Let $K$ be a convex body and $W$ be a subspace in $\mathbb{R}^n$. Then the function $g_{K,W}: W^\bot \rightarrow \mathbb{R}_+$ defined as $g_{K,W}(x) := \vol(K \cap (W + x))$ is log-concave on its support. 
\end{theorem}

\begin{theorem}[Property of log-concave density, Theorem 5.14 of~\cite{lv07}] \label{thm:high_dim_logconcave}
Let $f: \mathbb{R}^n \rightarrow \mathbb{R}_+$ be an isotropic log-concave density function. 
Then we have $f(x) \leq 2^{8n} n^{n/2}$ for every $x$. 
\end{theorem}

We also need the following result from~\cite{kls95}.

\begin{theorem}[Ellipsoidal approximation of convex body, \cite{kls95}] \label{thm:isotropic_rounding}
Let $K$ be an isotropic convex body in $\mathbb{R}^n$. Then,
\begin{align*}
\sqrt{\frac{n+1}{n}} \cdot B_2 \subseteq K \subseteq \sqrt{n(n+1)} \cdot B_2 ,
\end{align*}
where $B_2$ is the unit Euclidean ball in $\mathbb{R}^n$. 
\end{theorem} 

The following lemma is an immediate consequence of Theorem~\ref{thm:isotropic_rounding}.

\begin{lemma}[Stability of covariance] \label{lem:covariance_change}
Let $K$ be a convex body in $\mathbb{R}^n$ and $x \in K$ satisfies $\norm{x - \cg(K)}_{\Cov(K)^{-1}} \leq 0.1$. 
Let $H$ be a halfspace such that $x \in H$, then we have 
\begin{align*}
\frac{1}{5 n^2} \cdot \Cov(K) \preceq \Cov(K \cap H) \preceq n^2 \cdot \Cov(K) . 
\end{align*}
\end{lemma}

\begin{proof}
Without loss of generality, we may assume that $K$ is in isotropic position, in which case the condition that $\norm{x - \cg(K)}_{\Cov(K)^{-1}} \leq 0.1$ becomes $\norm{x}_2 \leq 0.1$. 
Theorem~\ref{thm:isotropic_rounding} then gives
\begin{align*} 
\sqrt{\frac{n+1}{n}} \cdot B_2 \subseteq K \subseteq \sqrt{n (n+1)} \cdot B_2 . 
\end{align*} 

Let halfspace $H_1$ be the translation of halfspace $H$ such that $x$ lies on its boundary hyperplane $H_1'$. Note that $K \cap H_1 \subseteq K \cap H$.  
Let $x' := \Pi_{H_1'}(\cg(K))$ be the orthogonal projection of $\cg(K) = 0$ onto the hyperplane $H_1'$. Then, 
\begin{align*}
\|x'\|_2 \leq \| x- 0\|_2 \leq 0.1 .
\end{align*}
This shows that the hyperplane $H_1'$ is at Euclidean distance at most $0.1$ from $0$. 
It then follows that $\sqrt{\frac{n+1}{n}} B_2 \cap H_1$ contains a ball of radius at least 
\begin{align*}
\frac{1}{2} \cdot \left(\sqrt{\frac{n+1}{n}} - 0.1 \right) \geq  0.45 \sqrt{\frac{n+1}{n}} \geq \sqrt{\frac{n+1}{5n}} ,
\end{align*}
where the last inequality uses $\sqrt{5} \times 0.45 \geq 1$. 
Since we have $\sqrt{\frac{n+1}{n}} B_2 \cap H_1 \subseteq K \cap H_1 \subseteq K \cap H$, this implies that $K \cap H$ contains a ball of radius $\sqrt{\frac{n+1}{5n}}$, and is contained in a ball of radius $\sqrt{n(n+1)}$. 
Consider the ellipsoid $E_{K \cap H} = \{y: y^\top \Cov(K \cap H)^{-1} y \leq 1 \}$. 
Then Theorem~\ref{thm:isotropic_rounding} implies that
\begin{align*}
\cg(K \cap H) + \sqrt{\frac{n+1}{n}} \cdot E_{K \cap H} \subseteq K \cap H \subseteq \cg(K \cap H) + \sqrt{n(n+1)} \cdot E_{K \cap H} .
\end{align*}
We thus have $\frac{1}{\sqrt{5} n} \cdot B_2 \subseteq E_{K \cap H} \subseteq n \cdot B_2$, and the statement of the lemma follows immediately.
\end{proof}

We note that some of these convex geometry tools have previously been used, for example, to find the densest sub-lattice in arbitary norm~\cite{dm13}.

\subsection{Cutting Plane Methods}

Cutting plane methods optimize a convex function $f$ by maintaining a convex set $K$ that contains the minimizer of $f$, which gets refined iteratively using the separating hyperplanes returned by the separation oracle. 
One of the most classical cutting plane methods is the center of gravity method, discovered independently by Levin \cite{l65} and Newman \cite{n65}.

\begin{algorithm}[htp!]\caption{}\label{alg:CenterOfGravity}
\begin{algorithmic}[1]
\Procedure{\textsc{CenterOfGravity}}{$\SO, K$} 
\State Query $\SO$ at $\cg(K)$
\If{$\SO$ outputs ``YES''}
\State \textbf{Return} ``YES''
\Else
\State Let $c$ be the output of $\SO$
\State \textbf{Return} $K' := K \cap \{x: c^\top x \geq c^\top \cg(K)\}$
\EndIf
\EndProcedure
\end{algorithmic}
\end{algorithm}

\begin{theorem}[Center of gravity method \cite{l65,n65}] \label{thm:CG}
Given a separation oracle $\SO$ for a convex function $f$ defined on $\mathbb{R}^n$ with minimizers $K^*$, and a convex body $K \subseteq \mathbb{R}^n$ containing $K^*$. If $\cg(K)$ doesn't minimize $f$, then the convex body $K'$ returned by $\text{CenterOfGravity}(\SO, K)$ above contains $K^*$ and satisfies $\vol(K') \leq (1-1/e) \cdot \vol(K)$. 
\end{theorem}

The center of gravity method is not efficient as it involves computing the centroid of convex bodies. 
Using sampling techniques to estimate $\cg(K)$ and $\Cov(K)$, an efficient implementation of the center of gravity method was given in \cite{bv04}. 
We start with the definition of $\epsilon$-approximate centroid and covariance. 

\begin{definition}[$\epsilon$-approximate centroid and covariance]
Let $0 < \epsilon < 1$ be a parameter. 
Given a convex body $K \subseteq \mathbb{R}^n$, we call $x_K \in \mathbb{R}^n$ an $\epsilon$-approximate centroid of $K$ if $\|x_K - \cg(K)\|_{\Cov(K)^{-1}} \leq \epsilon$. We call PSD matrix $\Sigma_K \in \mathbb{R}^{n \times n}$ an $\epsilon$-approximate covariance matrix if $(1-\epsilon) \cdot \Cov(K) \preceq \Sigma_K \preceq (1 + \epsilon) \cdot \Cov(K)$. 
\end{definition}

Constructing $\epsilon$-approximate centroids and covariance matrices via sampling for well-rounded convex bodies appeared in the works of \cite{kls97,alptj10,sv13}. The formulation of the following theorem is from \cite[Lemma 2.5 and Theorem 2.7]{jllv21} together with the standard fact that the uniform distribution over a convex body is log-concave.

\begin{theorem}[Approximate centroid and covariance by sampling, \cite{kls97,alptj10,sv13}] 
\label{thm:sampling_polytope}
Let parameters $0< \epsilon < 1$ and $0 < \delta < 1/2$. 
Given a convex body $K \subseteq \mathbb{R}^n$ specified by $m$ constraints, a point $x \in K$ and a PSD matrix $A \in \mathbb{R}^{n \times n}$ such that the following sandwiching condition holds 
\begin{align} \label{eq:warm_start}
x + E(A) \subseteq K \subseteq x + 2^{\poly(n)} \cdot E(A) ,
\end{align}
then there is a randomized algorithm that uses $m \cdot \poly(n, 1/\epsilon, \log(1/\delta))$ arithmetic operations to compute, with probability at least $1-\delta$, an $\epsilon$-approximate centroid $x_K$ and an $\epsilon$-approximate covariance matrix $\Sigma_K$ of $K$. 
\end{theorem}

Since approximate centroid and covariance matrix of a convex body give a sandwiching condition as in \eqref{eq:warm_start}, \cite{bv04} obtained the following efficient implementation of the center of gravity method. 
The theorem below comes from directly using \Cref{thm:sampling_polytope} in the algorithmic framework of \cite{bv04}. 

\begin{theorem}[Approximate center of gravity method, \cite{bv04}] \label{thm:randomwalkcg}
Let parameters $0< \epsilon < 0.01$ and $0 < \delta < 1/2$. 
Given a separation oracle $\SO$ for a convex function $f$ defined on $\mathbb{R}^n$ with minimizers $K^*$, a polytope $K$ with $m$ constraints containing $K^*$, an $\epsilon$-approximate centroid $x_K \notin K^*$ and an $\epsilon$-approximate covariance matrix $\Sigma_K$ of $K$, there exists a randomized algorithm $\textsf{RandomWalkCG}(\SO, K, x_K, \Sigma_K, \epsilon, \delta)$ that makes one call to $\SO$ and an extra $m \cdot \poly(n, 1/\epsilon, \log(1/\delta))$ arithmetic operations to return a polytope $K'$, a point $x_{K'} \in K'$ and a PSD matrix $\Sigma_{K'}$ such that the following hold with probability at least $1-\delta$: 
\begin{itemize}
	\item [(a)] $K^* \subseteq K'$ and $K'$ is the intersection of $K$ with a constraint output by $\SO$ at $x_K$,
	\item [(b)] $\vol(K') \leq \frac{2}{3} \cdot \vol(K)$,
	\item [(c)] $x_{K'}$ is an $\epsilon$-approximate centroid of $K'$, and
	\item [(d)] $\Sigma_{K'}$ is an $\epsilon$-approximate covariance matrix of $K'$. 
\end{itemize}
\end{theorem}

\section{Technical Lemmas}
\label{sec:technical_lemmas}

In this section, we prove a few technical lemmas which are key to our result.

\subsection{Dimension Reduction that Preserves Low-Complexity Rational Points}
\label{subsec:reduce_dimension}

Recall from \Cref{defn:Sphi} that $S_\varphi^n$ is the set of rational vectors with LCM vertex complexity at most $\varphi \geq 0$.

\begin{lemma}[Dimension reduction that preserves low-complexity rational points] \label{lem:reduce_dimension}
Given an affine subspace $W = x_0 + W_0$, where $W_0$ is a linear subspace of $\mathbb{R}^n$ and $x_0 \in \mathbb{R}^n$ is a fixed point, 
and an ellipsoid  $E = E(x_0, A)$ that has full rank on $W$. 
Given a vector $v \in \Pi_{W_0}(\mathbb{Z}^n) \setminus \{0\}$ with $\norm{v}_{A^{-1}} < 1/2^{2 \varphi + 1}$, where $\varphi \geq 0$ is an integer, then there exists a hyperplane $P \nsupseteq W$ such that $E \cap S_{\varphi}^n \subseteq P \cap W$. 
In particular, let $z \in \mathbb{Z}^n$ be such that $v = \Pi_{W_0}(z)$, then $P$ can be taken as 
\begin{align*}
P = \{x: v^\top x = (v - z)^\top x_0 + \lceil z^\top x_0 \rfloor_\varphi \} . 
\end{align*} 
\end{lemma}

\begin{proof}
Clearly we have $E \cap S_\varphi^n \subseteq W$ since $E \subseteq W$. 
It therefore suffices to show that the hyperplane $P$ given in the lemma statement satisfies $P \nsupseteq W$ and $E \cap S_\varphi^n \subseteq P$. 

Since $v \in W_0 \setminus \{0\}$ and $W_0$ is a translation of $W$, we have $P \nsupseteq W$. If $E \cap S_\varphi^n = \emptyset$, then the lemma statement trivially holds. We may therefore assume $E \cap S_\varphi^n \neq \emptyset$ in the following. Then for any rational vectors $x_1,x_2 \in E \cap S_\varphi^n$, we have
\begin{align*}
|v^\top (x_1 - x_2)| 
& \leq \norm{v}_{A^{-1}} \cdot \norm{x_1 - x_2}_A \\
& < \frac{1}{2^{2 \varphi + 1}} \cdot (\norm{x_1 - x_0}_A + \norm{x_2 - x_0}_A) \leq \frac{1}{2^{2 \varphi}} .
\end{align*}
Since $x_1,x_2 \in W \cap S_\varphi^n$, we have $x_1 - x_2 \in W_0 \cap S_{2 \varphi}$.
As $v = \Pi_{W_0}(z)$ where $z \in \mathbb{Z}^n$, we have 
\begin{align*}
v^\top (x_1 - x_2) = z^\top (x_1 - x_2) \in \mathbb{Z}/q ,
\end{align*}
for some positive integer $q \leq 2^{2 \varphi}$. 
It then follows that $v^\top x_1 = v^\top x_2$.  
Finally, we note that for any rational vector $x_1 \in E \cap S_\varphi^n$, we have
\begin{align*}
|z^\top (x_1 - x_0)| = | v^\top (x_1 - x_0) | \leq \norm{v}_{A^{-1}} \cdot \norm{x_1 - x_0}_A < \frac{1}{2^{2 \varphi + 1}} . 
\end{align*}
Since $z^\top x_1 \in \mathbb{Z}/q'$ for some $q' \leq 2^{\varphi}$, we have $z^\top x_1 = \lceil z^\top x_0 \rfloor_\varphi$.
Therefore, we have 
\begin{align*}
v^\top x_1 = \lceil z^\top x_0 \rfloor_{\varphi} + (v - z)^\top x_1  = \lceil z^\top x_0 \rfloor_{\varphi} + (v - z)^\top x_0 ,
\end{align*}  
where the last equality is because $v - z \in W_0^\bot$ and $x_1 - x_0 \in W_0$. 
This finishes the proof of the lemma. 
\end{proof}

We remark here that the rounding $\lceil \cdot \rfloor_\varphi$ in the construction of the hyperplane $P$ can be efficiently computed using the continued fraction method (e.g. \cite[Corollary 6.3a]{s98}.

\subsection{High Dimensional Slicing Lemma}

\begin{lemma}[High dimensional slicing lemma] \label{lem:high_dim_slicing}
Let $K$ be a convex body and $L$ be a full-rank lattice in $\mathbb{R}^n$. Let $W$ be an $(n-k)$-dimensional linear subspace of $\mathbb{R}^n$ such that $\dim(L \cap W) = n-k$. Then we have
\begin{align*}
\frac{\vol(K \cap W)}{\det(L \cap W)} \leq \frac{\vol(K)}{\det(L)} \cdot \frac{k^{O(k)}}{\lambda_1(L^*, K)^k} ,
\end{align*} 
where $L^*$ is the dual lattice, and $\lambda_1(L^*,K)$ is the shortest non-zero vector in $L^*$ under the norm $\norm{ \cdot }_{\Cov(K)}$.  
\end{lemma}

\begin{proof}
Note that $\vol(K \cap W) / \det(L \cap W)$, $\vol(K)/\det(L)$, and $\lambda_1(L^*, K)$ are preserved when applying the same linear transformation to $K$ and $L$ simultaneously. 
We can therefore rescale $K$ and $L$ such that $\Cov(K) = I$. 
We may further assume that $K \cap W \neq \emptyset$ as otherwise $\vol(K \cap W) = 0$ and the statement trivially holds.

We first upper bound $\vol(K \cap W)$ in terms of $\vol(K)$. 
To this end, we apply a translation on $K$ to obtain $K_0$ such that $\cg(K_0) = 0$, i.e. $K_0$ is in isotropic position, and it suffices to upper bound the cross-sectional volume $\vol(K_0 \cap (W + x))$ for an arbitrary $x \in W^\bot$. 
By identifying $W^\bot$ with $\mathbb{R}^k$, we note that the function $f(x)$ defined as $f(x) := \vol(K_0 \cap (W + x)) / \vol(K_0)$
is a log-concave density function on $\mathbb{R}^k$ by Brunn's principle (Theorem~\ref{thm:brunn}).
Furthermore, $f(x)$ is isotropic since $K_0$ is in isotropic position. 
It thus follows from Theorem~\ref{thm:high_dim_logconcave} that $f(x) \leq k^{O(k)}$, for any $x \in \mathbb{R}^k$. 
Note that $K = K_0 + \cg(K)$, we obtain from taking $x = -\cg(K)$ that
\begin{align} \label{eq:vol_ratio_bound}
\frac{\vol(K \cap W)}{\vol(K)} \leq k^{O(k)}.
\end{align}

We next upper bound $\det(L)$ in terms of $\det(L \cap W)$. 
Note that
\begin{align} \label{eq:det_ratio_bound_1}
\det(L) = \det(L \cap W) \cdot \det(\Pi_{W^\bot} (L)) = \frac{\det(L \cap W)}{\det(L^* \cap W^\bot)} ,
\end{align}
where the first equality follows from Fact~\ref{fact:lattice_projection}, and the second equality is due to Fact~\ref{fact:lattice_duality}. 
By Minkowski's first theorem (Theorem~\ref{thm:minkowski_shortest_vector}), we have
\begin{align*}
\lambda_1(L^*) \leq \lambda_1(L^* \cap W^\bot) \leq \sqrt{k} \cdot (\det(L^* \cap W^\bot))^{1/k}.
\end{align*}
Combine this with the earlier equation \eqref{eq:det_ratio_bound_1} gives 
\begin{align} \label{eq:lattice_ratio_bound}
\det(L) \leq \frac{\det(L \cap W) \cdot \sqrt{k}^k}{\lambda_1(L^*)^k}
\end{align}

It then follows from \eqref{eq:vol_ratio_bound} and \eqref{eq:lattice_ratio_bound} that  
\begin{align*}
\frac{\vol(K \cap W)}{\vol(K)} \cdot \frac{\det(L)}{\det(L \cap W)} \leq \frac{k^{O(k)}}{\lambda_1(L^*)^k} .
\end{align*}
This finishes the proof of the lemma.
\end{proof}

\section{Meta Algorithm}
\label{sec:meta_algo}

In this section, we present a simple meta algorithm (\Cref{alg:meta_algo}) that achieves the oracle complexity in \Cref{thm:polyhedra}. While this meta algorithm requires computing the centroids and covariance matrices of polytopes and is therefore not efficient, its oracle complexity analysis contains most of the key insights of this paper. 
We give an efficient (but more complicated) implementation of this meta algorithm and prove \Cref{thm:polyhedra} in \Cref{sec:implementation}. 

\begin{theorem}[Oracle Complexity in \Cref{thm:polyhedra}] \label{thm:meta_algo}
Given a separation oracle $\SO$ for a convex function $f$ defined on $\mathbb{R}^n$, and a $\gamma$-approximation algorithm $\textsc{ApproxSVP}$ for the shortest vector problem. 
If the set of minimizers $K^*$ of $f$ is a rational polyhedron contained in a box of radius $R$ and has LCM vertex complexity at most $\varphi \geq 0$, then there is a randomized algorithm that with high probability finds a vertex of $K^*$ using $O(n(\varphi + \log(\gamma n R)))$ calls to $\SO$.
\end{theorem}

\subsection{The Meta Algorithm}
\label{subsec:the_meta_algo}
By the argument in the beginning of \Cref{subsec:proof_overview}, we may assume without loss of generality that $f$ has a unique minimizer $x^* \in S_\varphi^n$. We therefore describe our algorithm under this assumption. 

Our meta algorithm maintains an affine subspace $W$, a polytope $K \subseteq W$ containing the rational minimizer $x^*$ of $f$, and a lattice $\Lambda$. 
It also maintains the centroid $x_K$ and covariance matrix $\Sigma_K$ of the polytope $K$. 
In the beginning, the affine subspace $W = \mathbb{R}^n$, polytope $K = B_\infty(R)$ and lattice $\Lambda = \mathbb{Z}^n$. 
In each iteration of the algorithm (i.e. each while loop), the algorithm uses the $\gamma$-approximation algorithm \textsc{ApproxSVP} to find a short non-zero vector $v \in \Lambda$ under $\Sigma_K$-norm. 
If the vector $v$ satisfies $\norm{v}_{\Sigma_K} \geq \frac{1}{10n2^{2 \varphi}}$,
then the algorithm runs the center of gravity method (\Cref{thm:CG}) for one more step, and updates $x_K$ and $\Sigma_K$ to be the centroid and covariance matrix of the new polytope $K$. 
We remark that the criterion for performing the cutting plane step comes from the convex geometry fact that $K \subseteq x_K + 2n \cdot E(\Sigma_K^{-1})$ (Theorem~\ref{thm:isotropic_rounding}).

If, on the other hand, that $\norm{v}_{\Sigma_K} < \frac{1}{10n2^{2 \varphi}}$, then the algorithm uses Lemma~\ref{lem:reduce_dimension} to find a hyperplane $P$ that contains $K \cap S_\varphi^n$, where we recall from \Cref{defn:Sphi} that $S_\varphi^n$ is the set of all rational vectors in $\mathbb{R}^n$ with LCM vertex complexity at most $\varphi$. 
Specifically, the hyperplane $P = \{x: v^\top x = (v - z)^\top x_K + \lceil z^\top x_K \rfloor_\varphi \}$ for some integral vector $z \in \mathbb{Z}^n$ such that $v = \Pi_{W_0}(z)$ and $W_0 = -x_K + W$ is the translation of $W$ that passes through the origin.
One may find such a vector $z \in \mathbb{Z}^n$ efficiently by solving the closest vector problem $\min_{z \in \mathbb{Z}^n} \norm{z - v}_{P_{W_0}}$, where $P_{W_0}$ is the projection matrix onto the subspace $W_0$. 
As mentioned earlier, the rounding $\lceil \cdot \rfloor_\varphi$ can also be performed efficiently using the continued fraction method. 
After constructing the hyperplane $P$, the algorithm then recurses on the lower-dimensional affine subspace $W \cap P$, updates $K$ to be $K \cap P$, and updates $x_K$ and $\Sigma_K$ to be the centroid and covariance matrix of the new polytope $K \cap P$. 
The algorithm obtains a new lattice with rank reduced by one by projecting the current lattice $\Lambda$ onto $P_0$, a translation of $P$ that passes through the origin.

The above procedure stops when $\dim(W) = 0$, in which case $K$ contains a unique rational point $x^*$ which will be the output of the algorithm. 
Note that when $\dim(W) = 1$, the algorithm reduces to a binary search on the segment $K \subseteq W$. 
A formal description of the algorithm is given in \Cref{alg:meta_algo}.

We remark that \Cref{alg:meta_algo} is not efficient since it requires the computation of the centroid and covariance matrix in Line \ref{line:cg_1} and \ref{line:cg_2}. Line \ref{line:cg_1} can easily be made efficient using the approximate center of gravity method as in \Cref{thm:randomwalkcg}. However, it is not clear how to efficiently implement Line \ref{line:cg_2} since we do not know an ellipsoid satisfying condition \eqref{eq:warm_start} in \Cref{thm:sampling_polytope}, and thus approximate centroid and covariance matrix might not be efficiently computable by sampling. We address this computational issue in the next section.

\begin{algorithm}[H]\caption{}\label{alg:meta_algo}
\begin{algorithmic}[1]
\Procedure{\textsc{MetaALG}}{$\SO,R, \varphi$} 
\State Affine subspace $W \leftarrow \mathbb{R}^n$, polytope $K \leftarrow B_\infty(R)$, lattice $\Lambda \leftarrow \mathbb{Z}^n$
\State Centroid $x_K \leftarrow \cg(K)$, covariance matrix $\Sigma_K \leftarrow \Cov(K)$ 
\Comment{$x_K + E(\Sigma_K^{-1}) / 2 \subseteq K \subseteq x_K + 2n \cdot E(\Sigma_K^{-1})$}
\While{$\dim(W) > 0$} 
	\State $v \leftarrow \textsc{ApproxSVP}(\Lambda, \Sigma_K)$ 
   \Comment{$v \in \Lambda \setminus \{0\}$}
	\If{$\norm{v}_{\Sigma_K} \geq \frac{1}{10n 2^{2 \varphi}}$}
		\State $K \leftarrow \textsc{CenterOfGravity}(\SO, K)$ \label{line:CG}
		\State $x_K \leftarrow \cg(K)$, $\Sigma_K \leftarrow \Cov(K)$ \label{line:cg_1}
	\Else
	\State Find $z \in \mathbb{Z}^n$ such that $v = \Pi_{W_0}(z)$ \label{line:dim_reduce_start_improved} 
	\Comment{Subspace $W_0 = -x_K + W$} 
	\State Construct $P \leftarrow \{y: v^\top y = (v - z)^\top x_K + \lceil z^\top x_K \rfloor_\varphi \}$ 
	\label{line:reduce_dim_begin}
	\State $W \leftarrow W \cap P$, $K \leftarrow K \cap P$ 
	\Comment{Dimension reduction}
	\State $x_K \leftarrow \cg(K)$, $\Sigma_K \leftarrow \Cov(K)$ \label{line:cg_2}
	\State Construct hyperplane $P_0 \leftarrow \{y : v^\top y = 0\}$
	\State $\Lambda \leftarrow \Pi_{P_0}(\Lambda)$ \label{line:reduce_dim_end} \Comment{Lattice projection}
	\EndIf
\EndWhile
\State \textbf{Return} unique point $x^* \in K$
\EndProcedure
\end{algorithmic}
\end{algorithm}

\subsection{Oracle Complexity Analysis}

We start by proving the correctness of \Cref{alg:meta_algo}.

\begin{lemma}[Correctness of \textsc{MetaALG}] \label{lem:correctness_meta_algo}
Assuming the conditions in \Cref{thm:meta_algo} and that $f$ has a unique minimizer $x^* \in S_\varphi^n$, \Cref{alg:meta_algo} finds $x^*$. 
\end{lemma}

\begin{proof}
Note that in the beginning of each iteration, we have  
$K \subseteq W$ and $\Lambda \subseteq W_0$, where $W_0$ is the translation of $W$ that passes through the origin. 
We first argue that the lattice $\Lambda$ is in fact the orthogonal projection of $\mathbb{Z}^n$ onto the subspace $W_0$, i.e. $\Lambda = \Pi_{W_0}(\mathbb{Z}^n)$. 
This is required for Lemma~\ref{lem:reduce_dimension} to be applicable. 
Clearly $\Lambda = \Pi_{W_0}(Z)$ holds in the beginning of the algorithm since $\Lambda = \mathbb{Z}^n$ and $W = \mathbb{R}^n$.
Notice that the \textsc{CenterOfGravity} procedure in Line~\ref{line:CG} keeps $\Lambda$ and $W$ the same. 
Each time we reduce the dimension in Line~\ref{line:reduce_dim_begin}-\ref{line:reduce_dim_end}, we have
\begin{align*}
\Pi_{W_0 \cap P_0}(\mathbb{Z}^n) = \Pi_{W_0 \cap P_0}(\Pi_{W_0}(\mathbb{Z}^n)) = \Pi_{W_0 \cap P_0}(\Lambda) ,
\end{align*}
where the first equality follows because $W_0 \cap P_0$ is a subspace of $W_0$.
Since $\Pi_{P_0}(\Lambda) = \Pi_{W_0 \cap P_0}(\Lambda)$ as $v \in W_0$, this shows that the invariant $\Lambda = \Pi_{W_0}(\mathbb{Z}^n)$ holds throughout the algorithm.

We now prove that Algorithm~\ref{alg:meta_algo} finds the unique minimizer $x^* \in S_\varphi^n$.  
Note that in the beginning of the algorithm, we have $x^* \in K$. 
Since \textsc{CenterOfGravity} in Line~\ref{line:CG} always preserves $x^* \in K$, we only need to prove that dimension reduction in Line~\ref{line:reduce_dim_begin}-\ref{line:reduce_dim_end} preserves $x^* \in K$.
In the following, we show the stronger statement that each dimension reduction iteration in Line~\ref{line:reduce_dim_begin}-\ref{line:reduce_dim_end} preserves all rational points in $K \cap S_\varphi^n$.

Since \Cref{alg:meta_algo} maintains $x_K = \cg(K)$ and $\Sigma_K = \Cov(K)$ in every iteration,
an immediate application of \Cref{thm:isotropic_rounding} gives the following sandwiching condition: 
\begin{align} \label{eq:sandwiching_conditon}
x_K + E(\Sigma_K^{-1})/2 \subseteq K \subseteq x_K + 2n \cdot E(\Sigma_K^{-1}).
\end{align}

Now we proceed to show that each dimension reduction iteration preserves all rational points in $K \cap S_\varphi^n$. 
By the RHS of \eqref{eq:sandwiching_conditon}, we have $K \cap S_\varphi^n \subseteq (x_K + 2n \cdot E(\Sigma_K^{-1})) \cap S_\varphi^n$. 
Since $\norm{v}_{\Sigma_K} < \frac{1}{10n2^{2 \varphi}}$ is satisfied in a dimension reduction iteration, Lemma~\ref{lem:reduce_dimension} shows that all rational points in $(x_K + 2n \cdot E(\Sigma_K^{-1})) \cap S_\varphi^n$ lie on the hyperplane given by $P = \{y: v^\top y = (v - z)^\top x_K + \lceil z^\top x_K \rfloor_\varphi \}$. 
Thus we have $K \cap S_\varphi^n \subseteq K \cap P$ and this finishes the proof of the lemma. 
\end{proof}

Next, we prove the oracle complexity upper bound of \Cref{alg:meta_algo} in \Cref{thm:meta_algo}. 

\begin{lemma}[Oracle complexity of \textsc{MetaALG}] \label{lem:oracle_complexity_meta_algo}
Assuming the conditions in 
\Cref{thm:meta_algo} and that $f$ has a unique minimizer $x^* \in S_\varphi^n$, \Cref{alg:meta_algo} makes at most $O(n (\varphi + \log(\gamma n R)))$ calls to $\SO$. 
\end{lemma}

\begin{proof}
We note that the oracle is only called when \textsc{CenterOfGravity} is invoked in Line~\ref{line:CG}, and each run of \textsc{CenterOfGravity} makes one call to $\SO$ according to \Cref{thm:CG}. 
To upper bound the total number of runs of \textsc{CenterOfGravity}, we consider the potential function 
\begin{align*}
\Phi = \log(\vol(K) \cdot \det(\Lambda)).
\end{align*}
In the beginning, $\Phi = \log(\vol(B_\infty(R)) \cdot \det(I)) = n \log (R)$. 
Each time \textsc{CenterOfGravity} is called in Line~\ref{line:CG}, we have from \Cref{thm:CG} that the volume of $K$ decreases by at least a constant factor, so the potential function decreases by at least $\Omega(1)$ additively.

To analyze the change in the potential function after dimension reduction, we consider a maximal sequence of consecutive dimension reduction iterations $t_0+1, \cdots, t_0 + k$, i.e. \textsc{CenterOfGravity} is invoked in iteration $t_0$ and $t_0 + k+1$, while every iteration in $t_0 + 1,\cdots, t_0 + k$ decreases the dimension by one. 
We shall use superscript $(i)$ to denote the corresponding notations in the beginning of iteration $t_0 + i$, for any integer $i \geq 0$. 
In particular, in the beginning of iteration $t_0 + 1$, we have a convex body $K^{(1)} \subseteq K^{(0)} \subseteq W^{(0)} = W^{(1)}$, and after the sequence of dimension reduction iterations, we reach a convex body $K^{(k + 1)} = K^{(1)} \cap W^{(k + 1)} \subseteq K^{(0)} \cap W^{(k+1)}$. 
The lattice changes from $\Lambda^{(0)} = \Lambda^{(1)} \subseteq W_0^{(1)}$ to $\Lambda^{(k + 1)} = \Pi_{W_0^{(k + 1)}}(\Lambda^{(1)}) = \Pi_{W_0^{(k + 1)}}(\Lambda^{(0)})$, where we recall that subspaces $W_0^{(i)}$ are translations of the affine subspaces $W^{(i)}$ that pass through the origin. 
Note that the potential at the beginning of this maximal sequence of dimension reduction iterations is 
\begin{align*}
e^{\Phi^{(0)}} = \vol(K^{(0)}) \cdot \det(\Lambda^{(0)}) = \frac{\vol(K^{(0)})}{\det( (\Lambda^{(0)})^*)} .
\end{align*}
The potential after this sequence of dimension reduction iterations is
\begin{align*}
e^{\Phi^{(k+1)} } & = \vol(K^{(k+1)}) \cdot \det(\Lambda^{(k+1)}) 
 = \vol(K^{(1)} \cap W^{(k+1)}) \cdot \det(\Pi_{W_0^{(k+1)}}(\Lambda^{(0)})) \\
& = \frac{\vol(K^{(1)} \cap W^{(k+1)})}{\det( (\Pi_{W_0^{(k+1)}}(\Lambda^{(0)}))^* )} = \frac{\vol(K^{(1)} \cap W^{(k+1)})}{\det( (\Lambda^{(0)})^* \cap W_0^{(k+1)} )} \leq \frac{\vol(K^{(0)} \cap W^{(k+1)})}{\det( (\Lambda^{(0)})^* \cap W_0^{(k+1)} )},
\end{align*}
where the last equality follows from the duality $(\Pi_{W_0^{(k+1)}}(\Lambda^{(0)}))^* = (\Lambda^{(0)})^* \cap W_0^{(k+1)}$ in Fact~\ref{fact:lattice_duality}.
Since $W^{(k+1)}$ is a translation of the subspace $W_0^{(k+1)}$, we can apply \Cref{lem:high_dim_slicing} by taking $L = (\Lambda^{(0)})^*$ to obtain

\begin{align} \label{eq:pot_change_dim_reduce_1}
e^{\Phi^{(k+1)} } \leq e^{\Phi^{(0)}} \cdot \frac{k^{O(k)}}{\lambda_1(\Lambda^{(0)}, K^{(0)})^k},
\end{align} 
where $\lambda_1(\Lambda^{(0)}, K^{(0)})$ is the shortest non-zero vector in $\Lambda^{(0)}$ under the norm $\norm{ \cdot}_{\Cov(K^{(0)})}$.
As \textsc{CenterOfGravity} is invoked in iteration $t_0$, we have $\norm{v^{(0)}}_{\Sigma_K^{(0)}} \geq \frac{1}{10n 2^{2 \varphi}}$ for the output vector $v^{(0)} \in \Lambda^{(0)} \setminus \{0\}$. 
Since the \textsc{ApproxSVP} procedure is $\gamma$-approximation and that $\Sigma_K^{(0)} = \Cov(K^{(0)})$, this implies that $\lambda_1(\Lambda^{(0)}, K^{(0)}) \geq \frac{\Omega(1)}{\gamma n 2^{2 \varphi}}$. It then follows that
\begin{align*}
e^{\Phi^{(k+1)} } \leq e^{\Phi^{(0)}} \cdot (\gamma n 2^{\varphi})^{O(k)}.
\end{align*}
This shows that after a sequence of $k$ dimension reduction iterations, the potential increases additively by at most $O(k \log(\gamma n 2^{\varphi}))$. 
As there are at most $n$ dimension reduction iterations, the total amount of potential increase due to dimension reduction iterations is thus at most $O(n \log(\gamma n 2^{\varphi}))$.

Finally we note that whenever the potential becomes smaller than $-10 n \log(20 n \gamma 2^{2 \varphi})$, Minkowski's first theorem (Theorem~\ref{thm:minkowski_shortest_vector}) shows the existence of a non-zero vector $v \in \Lambda$ with $\norm{v}_{\Sigma_K} < \frac{1}{20n \gamma 2^{2 \varphi}}$. 
This implies that the $\gamma$-approximation algorithm \textsc{ApproxSVP} for the shortest vector problem will find a non-zero vector $v' \in \Lambda$ that satisfies $\norm{v'}_{\Sigma_K} < \frac{1}{20n 2^{2 \varphi}}$, and thus such an iteration will not invoke $\textsc{CenterOfGravity}$. 
Therefore, \Cref{alg:meta_algo} runs $\textsc{CenterOfGravity}$ at most $O(n \log(\gamma n 2^{\varphi}) + n \log(R)) = O(n(\varphi + \log(\gamma n R)))$ times. 
Since each run of $\textsc{CenterOfGravity}$ makes one call to $\SO$, the total number of calls to $\SO$ made by \Cref{alg:meta_algo} is thus $O(n(\varphi + \log(\gamma n R)))$. 
This finishes the proof of the lemma. 
\end{proof}

\begin{proof}[Proof of \Cref{thm:meta_algo}]
By the argument in the beginning of \Cref{subsec:proof_overview}, we may assume without loss of generality that $f$ has a unique minimizer $x^* \in S_\varphi^n$.
The correctness of \Cref{alg:meta_algo} is given in \Cref{lem:correctness_meta_algo}, and its oracle complexity is upper bounded in \Cref{lem:oracle_complexity_meta_algo}. 
These finish the proof of the theorem. 
\end{proof}

\section{Efficient Implementation of the Meta Algorithm}

\label{sec:implementation}

In this section, we give an efficient implementation of \Cref{alg:meta_algo} from the previous section and prove \Cref{thm:polyhedra} which we restate below for convenience. 

\thmpolyhedra*

\subsection{The Efficient Implementation}

By the argument in the beginning of \Cref{subsec:proof_overview}, we may assume without loss of generality that $f$ has a unique minimizer $x^* \in S_\varphi^n$. For simplicity, we present our algorithm under this assumption. 


As mentioned in the last paragraph of \Cref{subsec:the_meta_algo}, we can efficiently implement Line \ref{line:cg_1} of \Cref{alg:meta_algo} by using the approximate center of gravity method in \Cref{thm:randomwalkcg}. We now address the issue of efficiently implementing Line \ref{line:cg_2} of \Cref{alg:meta_algo} in the following. 


To obtain an approximate centroid and covariance matrix of the polytope $K$ after dimension reduction, our efficient algorithm maintains two polytopes $K_{\SO} \subseteq K_{\free}$. The polytope $K_{\SO}$ plays the same role as $K$ in \Cref{alg:meta_algo}, and is the polytope formed by the separating hyperplanes from $\SO$. 
And $K_{\free}$ is a simple polytope for which we always know an approximate centroid $x_K$ and covariance matrix $\Sigma_K$. 
Our algorithm explicitly maintains the lists of constraints for the polytopes $K_{\SO}$ and $K_{\free}$ to efficiently perform computations on them.  
In particular, our algorithm can efficiently certify\footnote{In general, our algorithm might not be able to efficiently verify that the geometric objects $K_{\free}$ being the same as $K_{\SO}$. So whenever we say $K_{\free} = K_{\SO}$, we always mean it in the sense that it can be efficiently certified by checking that all constraints for $K_{\SO}$ appear in the list of constraints for $K_{\free}$.} that $K_{\free} = K_{\SO}$ when all the constraints for $K_{\SO}$ appear in the list of constraints for $K_{\free}$, since it is always maintained that $K_{\SO} \subseteq K_{\free}$.

In the beginning of the algorithm, $K_{\free} = K_{\SO}$ and we run $\textsc{RandomWalkCG}$ for both polytopes at the same time. 
When dimension reduction happens in Line \ref{line:dim_reduce_start_implementation}-\ref{line:dim_reduce_end_implementation}, $K_{\SO}$ is updated to be $K^{\new}_{\SO} = K_{\SO} \cap P$ and we no longer have approximations to $\cg(K^{\new}_{\SO})$ and $\Cov(K^{\new}_{\SO})$. 
To bypass this difficulty, our strategy is to update $K_{\free}$ to be a simple polytope $K^{\new}_{\free}$ containing $K^{\new}_{\SO}$ for which we know $\cg(K^{\new}_{\free})$ and $\Cov(K^{\new}_{\free})$, and ``learn'' $\cg(K^{\new}_{\SO})$ and $\Cov(K^{\new}_{\SO})$ by shrinking $K^{\new}_{\free}$ via $\textsc{RandomWalkCG}$ until it coincides with $K^{\new}_{\SO}$. 
Whenever $K^{\new}_{\free} = K^{\new}_{\SO}$ happens again (in the aforementioned sense that the constraints for $K^{\new}_{\SO}$ all appear in the list of constraints $K^{\new}_{\free}$), we have successfully learned an approximate centroid and covariance matrix of $K^{\new}_{\SO}$, and can continue to shrink $K^{\new}_{\SO}$ using $\textsc{RandomWalkCG}$ as before. 

Now we specify our choice of $K^{\new}_{\free}$ in the strategy above. 
Note that $K^\new_{\SO} \subseteq P \cap (x_K + 2n \cdot E(\Sigma_K^{-1}))$. 
Denoting the ellipsoid 
$P \cap (x_K + 2n \cdot E(\Sigma_K^{-1})) = E(w, A)$,
we can simply choose $K^{\new}_{\free}$ to be the smallest hyperrectangle containing $E(w, A)$, i.e. $K^{\new}_{\free} = w + A^{-1/2} B_\infty$, for which it is easy to compute an exact centroid and covariance matrix. 

Such choice of $K^{\new}_{\free}$ blows up the volume of the outer ellipsoid $P \cap (x_K + 2n \cdot E(\Sigma_K^{-1}))$ by a factor of $n^{O(n)}$, and thus shrinking $K^{\new}_{\free}$ seems to require much more $\SO$ calls. 
The crucial observation here is that when we shrink the volume of $K^{\new}_{\free}$, we do not need to make calls to $\SO$ since we already know the polytope $K^{\new}_{\SO} \subseteq K^{\new}_{\free}$. Instead, we simulate the separation oracle using the smaller polytope $K^{\new}_{\SO}$ via the procedure \textsc{FreeCG} (see \Cref{alg:freeCG}) until we have $K^{\new}_{\free} = K^{\new}_{\SO}$ again, at which point we regain approximations to $\cg(K^{\new}_{\SO})$ and $\Cov(K^{\new}_{\SO})$. 
If we are ever able to find a hyperplane $P^{\new}$ containing $K^{\new}_{\free} \cap S_\varphi^n$ even before reaching the point $K^{\new}_{\free} = K^{\new}_{\SO}$, we can further reduce the dimension. 
A formal description of the efficient implementation is given in \Cref{alg:implementation}.

\begin{algorithm}[htp!]\caption{}\label{alg:implementation}
\begin{algorithmic}[1]
\Procedure{\textsc{Main}}{$\SO,R, \varphi$} 
\State Affine subspace $W \leftarrow \mathbb{R}^n$, lattice $\Lambda \leftarrow \mathbb{Z}^n$
\State Polytopes $(K_{\free}, K_{\SO}) \leftarrow (B_\infty(R), B_\infty(R))$ \Comment{Maintain constraints explicitly for $K_{\free}$ and $K_{\SO}$}
\State $x_K \leftarrow \cg(K_{\free})$ and $\Sigma_K \leftarrow \Cov(K_{\free})$ 
\Comment{$x_K + E(\Sigma_K^{-1}) / 2 \subseteq K_{\free} \subseteq x_K + 2n \cdot E(\Sigma_K^{-1})$}
\State $\epsilon \leftarrow 0.01$, $\delta \leftarrow 1/\poly(n, \varphi, \log(\gamma R))$ \Comment{Parameters in Theorem~\ref{thm:randomwalkcg}}
\While{$\dim(W) > 0$} 
	\State $v \leftarrow \textsc{ApproxSVP}(\Lambda, \Sigma_K)$ 
   \Comment{$v \in \Lambda \setminus \{0\}$}
	\If{$\norm{v}_{\Sigma_K} \geq \frac{1}{10n2^{2 \varphi}}$}
		\If{$K_{\free} = K_{\SO}$} \Comment{List of constraints for $K_{\free}$ include that of $K_{\SO}$} 
		\State $(K', x_{K'}, \Sigma_{K'}) \leftarrow \textsc{RandomWalkCG}(\SO, K_{\free}, x_K, \Sigma_K, \epsilon, \delta)$ as in Theorem~\ref{thm:randomwalkcg} \label{line:randomwalkcg} 
		\State $(K_{\free}, K_{\SO}) \leftarrow (K', K')$, $x_K \leftarrow x_{K'}$, $\Sigma_K \leftarrow \Sigma_{K'}$ 
		\Else
		\State $(K_{\free}, x_K, \Sigma_K) \leftarrow \textsc{FreeCG}(K_{\free}, K_{\SO}, x_K, \Sigma_K)$ \label{line:freecg} \Comment{No $\SO$ call in this step}
		\EndIf
	\Else
	\State Find $z \in \mathbb{Z}^n$ such that $v = \Pi_{W_0}(z)$ \label{line:dim_reduce_start_implementation} 
	\Comment{Subspace $W_0 = -x_K + W$} 
	\State Hyperplane $P \leftarrow \{y: v^\top y = (v - z)^\top x_K + \lceil z^\top x_K \rfloor_\varphi \}$  
	\State $W \leftarrow W \cap P$, $K_{\SO} \leftarrow K_{\SO} \cap P$ \Comment{Dimension reduction} 
	\State $K_{\free} \leftarrow w + A^{-1/2} B_\infty$ \Comment{Ellipsoid $E(w, A) := P \cap (x_K + 2n \cdot E(\Sigma_K^{-1}))$} 
	\State $x_K \leftarrow \cg(K_{\free})$, $\Sigma_K \leftarrow \Cov(K_{\free})$
	\State Hyperplane $P_0 \leftarrow \{y : v^\top y = 0\}$, lattice $\Lambda \leftarrow \Pi_{P_0}(\Lambda)$  \Comment{Lattice projection} 
	\label{line:dim_reduce_end_implementation}
	\EndIf
\EndWhile
\State \textbf{Return} unique point $x^* \in K_{\SO}$
\EndProcedure
\end{algorithmic}
\end{algorithm}

\begin{algorithm}[htp!]\caption{}\label{alg:freeCG}
\begin{algorithmic}[1]
\Procedure{\textsc{FreeCG}}{$K_{\free},K_{\SO}, x_K, \Sigma_K$} 
\If{$x_K \notin K_{\SO}$} \Comment{Check the constraints for $K_{\SO}$} 
	\State Find constraint $a^\top x \leq b$ of $K_{\SO}$ violated by $x_K$
	\State $H \leftarrow \{x:a^\top x \leq a^\top x_K\}$ \Comment{$x_K$ lies on the boundary of $H$}
	\State $K'_{\free} \leftarrow K_{\free} \cap H$ \Comment{Volume of $K_{\free}$ shrinks}
	\State Obtain $\epsilon$-approx. centroid $x_{K'}$ and cov. $\Sigma_{K'}$ of $K'_{\free}$ as in \Cref{thm:randomwalkcg} 
\Else
	\State Find any constraint $H = \{x:a^\top x \leq b\}$ of $K_{\SO}$ that is not a constraint of $K_{\free}$ \Comment{$x_K \in H$}
	\State $K'_{\free} \leftarrow K_{\free} \cap H$ \Comment{$K_{\free}$ learns one more constraint of $K_{\SO}$}
	\State Obtain $\epsilon$-approx. centroid $x_{K'}$ and cov. $\Sigma_{K'}$ of $K'_{\free}$ as in \Cref{thm:sampling_polytope} \Comment{Validity by \Cref{lem:covariance_change}}
\EndIf
\State \textbf{Return} $K'_{\free}, x_{K'}, \Sigma_{K'}$
\EndProcedure
\end{algorithmic}
\end{algorithm}

\subsection{Proof of Main Result}

By the argument in the beginning of \Cref{subsec:proof_overview}, we can assume wlog that $f$ has a unique minimizer $x^* \in S_\varphi^n$.  
We first prove the correctness and oracle complexity of \Cref{alg:implementation}.
These proofs are very similar to the proofs of \Cref{lem:correctness_meta_algo} and \ref{lem:oracle_complexity_meta_algo} from the previous section, so we only highlight the differences.

\begin{lemma}[Correctness of \textsc{Main}] \label{lem:correctness_implementation}
Assuming the conditions in \Cref{thm:polyhedra} and that $f$ has a unique minimizer $x^* \in S_\varphi^n$, \Cref{alg:implementation} finds $x^*$. 
\end{lemma}

\begin{proof}
As in the proof of \Cref{lem:correctness_meta_algo}, we only need to verify that $x^* \in K_{\SO}$ is preserved under dimension reduction in Line \ref{line:dim_reduce_start_implementation}-\ref{line:dim_reduce_end_implementation}.
Let's assume that $x^* \in K_{\SO}$ before dimension reduction. 
Since \Cref{thm:randomwalkcg} guarantees $\|x_K - \cg(K_{\free})\|_{(\Sigma_K)^{-1}} \leq \epsilon$ and $(1-\epsilon) \cdot \Cov(K_{\free}) \preceq \Sigma_K \preceq (1+\epsilon) \cdot \Cov(K_{\free})$ with $\epsilon = 0.01$, it follows from \Cref{thm:isotropic_rounding} that \eqref{eq:sandwiching_conditon} still holds with $K$ replaced by $K_{\free}$: 
\begin{align*}
x_K + E(\Sigma_K^{-1})/2 \subseteq K_{\free} \subseteq x_K + 2n \cdot E(\Sigma_K^{-1}).
\end{align*}

Proceeding from here, the same argument as in the proof of \Cref{lem:correctness_meta_algo} shows that $K_{\free} \cap S_\varphi^n \subseteq P$. 
Also note that \Cref{alg:implementation} always maintains $K_{\SO} \subseteq K_{\free}$.
It follows that $K_{\SO} \cap S_{\varphi}^n \subseteq K_{\free} \cap S_\varphi^n \subseteq P$, i.e. all rational points in $K_{\SO} \cap S_{\varphi}^n$ are preserved during dimension reduction. This implies that $x^* \in K_{\SO} \cap P$ after dimension reduction and completes the proof of the lemma. 
\end{proof}

\begin{lemma}[Oracle complexity of \textsc{Main}] \label{lem:oracle_complexity_implementation}
Assuming the conditions in 
\Cref{thm:polyhedra} and that $f$ has a unique minimizer $x^* \in S_\varphi^n$, \Cref{alg:implementation} makes at most $O(n (\varphi + \log(\gamma n R)))$ calls to the separation oracle $\SO$ with high probability. 
\end{lemma}

\begin{proof}
Note that \Cref{alg:implementation} always maintains $K_{\SO} \subseteq K_{\free}$, and $\SO$ is only called in Line \ref{line:randomwalkcg} when $K_{\SO} = K_{\free}$. 
Since each run of \textsc{RandomWalkCG} in Line \ref{line:randomwalkcg} succeeds with probability $\delta = 1/\poly(n, \varphi, \log(\gamma R))$ for a large enough polynomial by \Cref{thm:randomwalkcg}, union bound implies that with high probability, the first $O(n (\varphi + \log(\gamma n R)))$ run of \textsc{RandomWalkCG} in Line \ref{line:randomwalkcg} all succeed.
We condition on this event. 
Then applying exactly the same analysis as in the proof of \Cref{lem:oracle_complexity_meta_algo} to the potential function
\begin{align*}
\Phi_{\SO} := \log(\vol(K_{\SO}) \cdot \det(\Lambda))
\end{align*}
gives the oracle complexity bound in the lemma. 
\end{proof}

Next, we show that \Cref{alg:implementation} makes at most $\poly(n, \varphi, \log(\gamma R))$ calls to $\textsc{FreeCG}$ with high probability. Since each call to $\textsc{FreeCG}$ can be implemented in $\poly(n, \varphi, \log(\gamma R))$ time by checking all the constraints of $K_{\SO}$, this will imply the bound on the number of arithmetic operations in \Cref{thm:polyhedra}.

\begin{lemma}[Number of \textsc{FreeCG} calls] \label{lem:freecg_steps}
Assuming the conditions in \Cref{thm:polyhedra} and that $f$ has a unique minimizer $x^* \in S_\varphi^n$, \Cref{alg:implementation} makes at most $\poly(n, \varphi, \log(\gamma R))$ calls to $\textsc{FreeCG}$ with high probability. 
\end{lemma}

\begin{proof}
As in the proof above, we condition on the high probability event that the first $\poly(n, \varphi, \log(\gamma R))$ calls to \textsc{RandomWalkCG} as well as the sampling algorithm in \Cref{thm:sampling_polytope} all succeed.
In the beginning of the algorithm, $K_{\SO} = B_{\infty}(R)$ and thus can be specified using $2n$ constaints. 
An additional constraint is placed on $K_{\SO}$ each time $\SO$ is called, and since the number of $\SO$ calls is at most $O(n(\varphi + \log(\gamma n R))$, the number of constraints \Cref{alg:implementation} maintains for the specification of $K_{\SO}$ can be at most $O(n(\varphi + \log(\gamma n R))$ throughout. 

Now we upper bound the number of calls to $\textsc{FreeCG}$. In fact, we show that the total number of cutting plane steps for $K_{\free}$ in Line \ref{line:randomwalkcg} and \ref{line:freecg} of \Cref{alg:implementation} is at most $\poly(n,\varphi, \log(\gamma R))$. Our strategy is to consider the potential function
\begin{align*}
\Phi_{\free} := \log(\vol(K_{\free}) \cdot \det(\Lambda)) ,
\end{align*}
and repeat the analysis as in the proof of \Cref{lem:oracle_complexity_meta_algo}. However, there are two main differences that we highlight below. 

The first main difference is that when we reduce the dimension in Line \ref{line:dim_reduce_start_implementation}-\ref{line:dim_reduce_end_implementation} of \Cref{alg:implementation}, we are not simply slicing $K_{\free}$ by the hyperplane $P$. Instead, we first replace $K_{\free}$ by its outer containing ellipsoid $x_K + 2n \cdot E(\Sigma_K^{-1})$, then further replace the sliced ellipsoid $E(w, A) = P \cap (x_K + 2n \cdot E(\Sigma_K^{-1}))$ by its outer containing hyperrectangle $K^{\new}_{\free} := w + A^{-1/2} B_\infty$. 
Since we have the sandwiching condition that
\begin{align*}
x_K + E(\Sigma_K^{-1})/2 \subseteq K_{\free} \subseteq x_K + 2n \cdot E(\Sigma_K^{-1}),
\end{align*}
replacing $K_{\free}$ by $x_K + 2n \cdot E(\Sigma_K^{-1})$ increases its volume by at most $n^{O(n)}$. Also note that replacing an ellipsoid by its outer containing hyperrectangle increases its volume by at most $n^{O(n)}$. 
It then follows that these replacements contribute to at most a factor of $n^{O(n)}$ to $\vol(K_{\free})$ for each dimension reduction step. As there are at most $n$ dimension reduction steps, the increase in $\Phi_{\free}$ due to these replacements is at most $O(n^2 \log(n))$ additively. 

The second main difference is that not every call to \textsc{FreeCG} decreases $\vol(K_{\free})$ by a constant factor. In particular, this is the case if $x_K \in K_{\SO}$ in \Cref{alg:freeCG} and we add to $K_{\free}$ one constraint of $K_{\SO}$ that is currently not a constraint of $K_{\free}$.
However, since we have shown above that $K_{\SO}$ has at most $O(n (\varphi + \log(\gamma n R)))$ constraints, this case can happen at most $O(n (\varphi + \log(\gamma n R)))$ in each dimension until all the constraints for $K_{\SO}$ appear in the list of constraints for $K_{\free}$, in which case our algorithm can efficiently certify that $K_{\free} = K_{\SO}$. 
Whenever this happens, no additional call to \textsc{FreeCG} will happen until the dimension is further reduced. 

Incorporating the above two differences into the analysis as in the proof of  \Cref{lem:oracle_complexity_meta_algo}, we obtain that the total number of cutting plane steps in Line \ref{line:randomwalkcg} and \ref{line:freecg} applied to $K_{\free}$ is at most $O(n^2(\varphi + \log (\gamma n R)))$. This is also an upper bound on the number of calls to \textsc{FreeCG}, and thus proves the lemma. 
\end{proof}

\begin{proof}[Proof of \Cref{thm:polyhedra}]
By the argument in the beginning of \Cref{subsec:proof_overview}, we may assume without loss of generality that $f$ has a unique minimizer $x^* \in S_\varphi^n$. 
The correctness of \Cref{alg:implementation} is given in \Cref{lem:correctness_implementation}, and its oracle complexity is upper bounded in \Cref{lem:oracle_complexity_implementation}. 
We are thus left to upper bound the total number of arithmetic operations used by \Cref{alg:implementation}.

By \Cref{lem:freecg_steps}, \Cref{alg:implementation} makes at most $\poly(n, \varphi, \log(\gamma R))$ calls to \textsc{FreeCG} and each such step can be implemented using $\poly(n, \varphi, \log(\gamma R))$ arithmetic operations. 
Since $\textsc{ApproxSVP}$ is called after each cutting plane step in Line \ref{line:randomwalkcg} and \ref{line:freecg}, the total number of calls to $\textsc{ApproxSVP}$ is at most $\poly(n, \varphi, \log(\gamma R))$. 
Note that the remaining part of the algorithm takes $\poly(n, \varphi, \log(\gamma R))$ arithmetic operations. 
This gives the upper bound on the number of arithmetic operations
and finishes the proof of the theorem. 
\end{proof}


\section{Submodular Function Minimization}

\label{sec:submodular_appendix}

In this section, we do not seek to give a comprehensive introduction to submodular functions, but only provide the necessary definitions and properties that are needed for the proof of Theorem~\ref{thm:submodular_main}. 
We refer interested readers to the famous textbook by Schrijver~\cite{s03} or the extensive survey by McCormick~\cite{m05} for more details on submodular functions.

\subsection{Preliminaries}

Throughout this section, we use $[n] = \{1,\cdots, n\}$ to denote the ground set and let $f: 2^{[n]} \rightarrow \mathbb{Z}$ be a set function defined on subsets of $[n]$. 
For a subset $S \subseteq [n]$ and an element $i \in [n]$, we define $S + i := S \cup \{i\}$.  
A set function $f$ is {\em submodular} if it satisfies the following property of {\em diminishing marginal differences}: 

\begin{definition}[Submodularity] A function $f: 2^{[n]} \rightarrow \mathbb{Z}$ is submodular if $f(T + i) - f(T) \leq f(S + i) - f(S)$, for any subsets $S \subseteq T \subseteq [n]$ and $i \in [n] \setminus T$. 
\end{definition}

Throughout this section, the set function $f$ we work with is assumed to be submodular even when it is not stated explicitly. 
We may assume without loss of generality that $f(\emptyset) = 0$ by replacing $f(S)$ by $f(S) - f(\emptyset)$. 
We assume that $f$ is accessed by an {\em evaluation oracle}, and use $\EO$ to denote the time to compute $f(S)$ for a subset $S$. 
Our algorithm for SFM is based on a standard convex relaxation of a submodular function, known as the Lov\'asz extension~\cite{gls88}.  

\begin{definition}[Lov\'asz extension]
The Lov\'asz extension $\hat{f}:[0,1]^n \rightarrow \mathbb{R}$ of a submodular function $f$ is defined as 
\begin{align*}
\hat{f}(x) = \E_{t \sim [0,1]} [f(\{i: x_i \geq t\})],
\end{align*}
where $t \sim [0,1]$ is drawn uniformly at random from $[0,1]$. 
\end{definition}

The Lov\'asz extension $\hat{f}$ of a submodular function $f$ has many desirable properties. 
In particular, $\hat{f}$ is a convex relaxation of $f$ and it can be evaluated efficiently. 

\begin{theorem}[Properties of Lov\'asz extension]
\label{thm:lovasz_extension_properties}
Let $f: 2^{[n]} \rightarrow \mathbb{Z}$ be a submodular function and $\hat{f}$ be its Lov\'asz extension. Then, 
\begin{itemize}
	\item [(a)] $\hat{f}$ is convex and $\min_{x \in [0,1]^n} \hat{f}(x) = \min_{S \subseteq [n]} f(S)$;
	\item [(b)] $f(S) = \hat{f}(I_S)$ for any subset $S \subseteq [n]$, where $I_S$ is the indicator vector for $S$; 
	\item [(c)] Suppose $x \in [0,1]^n$ satisfies $x_1 \geq \cdots \geq x_n$, then $\hat{f}(x) = \sum_{i=1}^n (f([i]) - f([i-1])) x_i$;
	\item [(d)] The set of minimizers of $\hat{f}$ is the convex hull of the set of minimizers of $f$. 
\end{itemize}
\end{theorem}

Next we address the question of implementing the separation oracle (as in Definition~\ref{def:separation_oracle}) using the evaluation oracle of $f$.

\begin{theorem}[Separation oracle implementation for Lov\'asz extension, Theorem 61 of~\cite{lsw15}] \label{thm:SO_from_EO}
Let $f: 2^{[n]} \rightarrow \mathbb{Z}$ be a submodular function and $\hat{f}$ be its Lov\'asz extension, then a separation oracle for $\hat{f}$ can be implemented in time $O(n \cdot \EO + n^2)$. 
\end{theorem}

\subsection{Proof of Theorem~\ref{thm:submodular_main}}

Before presenting the proof, we restate Theorem~\ref{thm:submodular_main} for convenience. 

\Submodular*

\begin{proof}
We apply Corollary~\ref{cor:instantiation} to the Lov\'asz extension $\hat{f}$ of the submodular function $f$ with $R = 1$. 
By part (a) and (d) of Theorem~\ref{thm:lovasz_extension_properties}, $\hat{f}$ is a convex function that satisfies the assumption $(\star)$ in Corollary~\ref{cor:instantiation}
Thus Corollary~\ref{cor:instantiation} gives a strongly polynomial algorithm for finding an integral minimizer of $\hat{f}$ that makes $O(n^2 \log\log (n)/\log (n))$ calls to a separation oracle of $\hat{f}$, and an exponential time algorithm that finds an integral minimizer of $\hat{f}$ using $O(n \log(n))$ separation oracle calls. 
This integral minimizer also gives a minimizer of $f$. 
Since a separation oracle for $\hat{f}$ can be implemented using $O(n)$ calls to $\EO$ by Theorem~\ref{thm:SO_from_EO}, the total number of calls to the evaluation oracle is thus $O(n^3 \log \log (n)/\log (n))$ for the strongly polynomial algorithm, and is $O(n^2 \log(n))$ for the exponential time algorithm. 
This proves the theorem. 
\end{proof}



\section*{Acknowledgments}
I would like to thank the anonymous referees of Journal of the ACM for very insightful comments. 
I thank my advisor Yin Tat Lee for advising this project.
Part of this work is inspired from earlier notes by Yin Tat Lee and Zhao Song.   
A special thanks to Daniel Dadush for pointing out the implication of the Gr{\"o}tschel-Lov{\'a}sz-Schrijver approach to our problem, suggesting the high dimensional slicing lemma which greatly simplifies my earlier proofs, and to Daniel Dadush and Thomas Rothvoss for pointing out that our framework implies $O(n^2 \log(n))$ oracle complexity for SFM by solving SVP exactly. 
I also thank Thomas Rothvoss for other useful comments and his wonderful lecture notes on integer optimization and lattice theory.   
I also thank Jonathan Kelner, Janardhan Kulkarni, Aaron Sidford, Zhao Song, Santosh Vempala, and Sam Chiu-wai Wong for helpful discussions on this project. 

\bibliographystyle{alpha}
\bibliography{bib.bib}

\begin{thebibliography}{ADRSD15}

\bibitem[AC91]{ac91}
Ilan Adler and Steven Cosares.
\newblock A strongly polynomial algorithm for a special class of linear
  programs.
\newblock {\em Operations Research}, 39(6):955--960, 1991.

\bibitem[ADRSD15]{adrs15}
Divesh Aggarwal, Daniel Dadush, Oded Regev, and Noah Stephens-Davidowitz.
\newblock Solving the shortest vector problem in 2n time using discrete
  gaussian sampling.
\newblock In {\em Proceedings of the forty-seventh annual ACM symposium on
  Theory of computing}, pages 733--742, 2015.

\bibitem[AKS01]{aks01}
Mikl{\'o}s Ajtai, Ravi Kumar, and Dandapani Sivakumar.
\newblock A sieve algorithm for the shortest lattice vector problem.
\newblock In {\em Proceedings of the thirty-third annual ACM symposium on
  Theory of computing}, pages 601--610, 2001.

\bibitem[ALPTJ10]{alptj10}
Rados{\l}aw Adamczak, Alexander Litvak, Alain Pajor, and Nicole
  Tomczak-Jaegermann.
\newblock Quantitative estimates of the convergence of the empirical covariance
  matrix in log-concave ensembles.
\newblock {\em Journal of the American Mathematical Society}, 23(2):535--561,
  2010.

\bibitem[BV04]{bv04}
Dimitris Bertsimas and Santosh Vempala.
\newblock Solving convex programs by random walks.
\newblock {\em Journal of the ACM (JACM)}, 51(4):540--556, 2004.

\bibitem[Cas71]{c71}
John William~Scott Cassels.
\newblock {\em An introduction to the theory of numbers}.
\newblock Springer-Verlag, 1971.

\bibitem[Chu12]{c12}
Sergei Chubanov.
\newblock A strongly polynomial algorithm for linear systems having a binary
  solution.
\newblock {\em Mathematical programming}, 134(2):533--570, 2012.

\bibitem[Chu15]{c15}
Sergei Chubanov.
\newblock A polynomial algorithm for linear optimization which is strongly
  polynomial under certain conditions on optimal solutions, 2015.

\bibitem[CM94]{cm94}
Edith Cohen and Nimrod Megiddo.
\newblock Improved algorithms for linear inequalities with two variables per
  inequality.
\newblock {\em SIAM Journal on Computing}, 23(6):1313--1347, 1994.

\bibitem[Dad12]{d12thesis}
Daniel Dadush.
\newblock {\em Integer programming, lattice algorithms, and deterministic
  volume estimation}.
\newblock PhD thesis, Georgia Institute of Technology, 2012.

\bibitem[DHNV20]{dhnv20}
Daniel Dadush, Sophie Huiberts, Bento Natura, and L{\'a}szl{\'o}~A V{\'e}gh.
\newblock A scaling-invariant algorithm for linear programming whose running
  time depends only on the constraint matrix.
\newblock In {\em Proceedings of the 52nd Annual ACM SIGACT Symposium on Theory
  of Computing}, pages 761--774, 2020.

\bibitem[DM13]{dm13}
Daniel Dadush and Daniele Micciancio.
\newblock Algorithms for the densest sub-lattice problem.
\newblock In {\em Proceedings of the twenty-fourth annual ACM-SIAM symposium on
  Discrete algorithms}, pages 1103--1122. SIAM, 2013.

\bibitem[DVZ18]{dvz18}
Daniel Dadush, L{\'a}szl{\'o}~A V{\'e}gh, and Giacomo Zambelli.
\newblock Geometric rescaling algorithms for submodular function minimization.
\newblock In {\em Proceedings of the Twenty-Ninth Annual ACM-SIAM Symposium on
  Discrete Algorithms}, pages 832--848. SIAM, 2018.

\bibitem[DVZ20]{dvz20}
Daniel Dadush, L{\'a}szl{\'o}~A V{\'e}gh, and Giacomo Zambelli.
\newblock Rescaling algorithms for linear conic feasibility.
\newblock {\em Mathematics of Operations Research}, 45(2):732--754, 2020.

\bibitem[Edm70]{e70}
Jack Edmonds.
\newblock Submodular functions, matroids, and certain polyhedra.
\newblock {\em Edited by G. Goos, J. Hartmanis, and J. van Leeuwen}, page~11,
  1970.

\bibitem[FI03]{fi03}
Lisa Fleischer and Satoru Iwata.
\newblock A push-relabel framework for submodular function minimization and
  applications to parametric optimization.
\newblock {\em Discrete Applied Mathematics}, 131(2):311--322, 2003.

\bibitem[GLS81]{gls81}
Martin Gr{\"o}tschel, L{\'a}szl{\'o} Lov{\'a}sz, and Alexander Schrijver.
\newblock The ellipsoid method and its consequences in combinatorial
  optimization.
\newblock {\em Combinatorica}, 1(2):169--197, 1981.

\bibitem[GLS84]{gls84}
Martin Gr{\"o}tschel, L{\'a}szl{\'o} Lov{\'a}sz, and Alexander Schrijver.
\newblock Geometric methods in combinatorial optimization.
\newblock In {\em Progress in combinatorial optimization}, pages 167--183.
  Elsevier, 1984.

\bibitem[GLS88]{gls88}
Martin Gr{\"o}tschel, L{\'a}szl{\'o} Lov{\'a}sz, and Alexander Schrijver.
\newblock {\em Geometric algorithms and combinatorial optimization}.
\newblock Springer, 1988.

\bibitem[IFF01]{iff01}
Satoru Iwata, Lisa Fleischer, and Satoru Fujishige.
\newblock A combinatorial strongly polynomial algorithm for minimizing
  submodular functions.
\newblock {\em Journal of the ACM (JACM)}, 48(4):761--777, 2001.

\bibitem[IO09]{io09}
Satoru Iwata and James~B Orlin.
\newblock A simple combinatorial algorithm for submodular function
  minimization.
\newblock In {\em Proceedings of the twentieth annual ACM-SIAM symposium on
  Discrete algorithms}, pages 1230--1237. SIAM, 2009.

\bibitem[Iwa03]{i03}
Satoru Iwata.
\newblock A faster scaling algorithm for minimizing submodular functions.
\newblock {\em SIAM Journal on Computing}, 32(4):833--840, 2003.

\bibitem[Iwa08]{i08}
Satoru Iwata.
\newblock Submodular function minimization.
\newblock {\em Mathematical Programming}, 112(1):45, 2008.

\bibitem[JLLV21]{jllv21}
He~Jia, Aditi Laddha, Yin~Tat Lee, and Santosh Vempala.
\newblock Reducing isotropy and volume to kls: an o*(n 3 $\psi$ 2) volume
  algorithm.
\newblock In {\em Proceedings of the 53rd Annual ACM SIGACT Symposium on Theory
  of Computing}, pages 961--974, 2021.

\bibitem[Kha80]{k80}
Leonid~G Khachiyan.
\newblock Polynomial algorithms in linear programming.
\newblock {\em USSR Computational Mathematics and Mathematical Physics},
  20(1):53--72, 1980.

\bibitem[KLS95]{kls95}
Ravi Kannan, L{\'a}szl{\'o} Lov{\'a}sz, and Mikl{\'o}s Simonovits.
\newblock Isoperimetric problems for convex bodies and a localization lemma.
\newblock {\em Discrete \& Computational Geometry}, 13(3-4):541--559, 1995.

\bibitem[KLS97]{kls97}
Ravi Kannan, L{\'a}szl{\'o} Lov{\'a}sz, and Mikl{\'o}s Simonovits.
\newblock Random walks and an o*(n5) volume algorithm for convex bodies.
\newblock {\em Random Structures \& Algorithms}, 11(1):1--50, 1997.

\bibitem[KTE88]{kte88}
Leonid~G Khachiyan, Sergei~Pavlovich Tarasov, and I.~I. Erlikh.
\newblock The method of inscribed ellipsoids.
\newblock In {\em Soviet Math. Dokl}, volume~37, pages 226--230, 1988.

\bibitem[Lev65]{l65}
Anatoly~Yur'evich Levin.
\newblock An algorithm for minimizing convex functions.
\newblock In {\em Doklady Akademii Nauk}, volume 160, pages 1244--1247. Russian
  Academy of Sciences, 1965.

\bibitem[LLL82]{lll82}
Arjen Lenstra, Hendrik Lenstra, and L{\'a}szl{\'o} Lov{\'a}sz.
\newblock Factoring polynomials with rational coefficients.
\newblock {\em Math. Ann}, 261:515--534, 1982.

\bibitem[LSW15]{lsw15}
Yin~Tat Lee, Aaron Sidford, and Sam Chiu-wai Wong.
\newblock A faster cutting plane method and its implications for combinatorial
  and convex optimization.
\newblock In {\em 2015 IEEE 56th Annual Symposium on Foundations of Computer
  Science}, pages 1049--1065. IEEE, 2015.

\bibitem[LV07]{lv07}
L{\'a}szl{\'o} Lov{\'a}sz and Santosh Vempala.
\newblock The geometry of logconcave functions and sampling algorithms.
\newblock {\em Random Structures \& Algorithms}, 30(3):307--358, 2007.

\bibitem[McC05]{m05}
S~Thomas McCormick.
\newblock Submodular function minimization.
\newblock {\em Discrete Optimization}, 12:321--391, 2005.

\bibitem[Meg83]{m83}
Nimrod Megiddo.
\newblock Towards a genuinely polynomial algorithm for linear programming.
\newblock {\em SIAM Journal on Computing}, 12(2):347--353, 1983.

\bibitem[Min53]{m53}
Hermann Minkowski.
\newblock Geometrie der zahlen.
\newblock {\em Chelsea, reprint}, 1953.

\bibitem[MV13]{mv13}
Daniele Micciancio and Panagiotis Voulgaris.
\newblock A deterministic single exponential time algorithm for most lattice
  problems based on voronoi cell computations.
\newblock {\em SIAM Journal on Computing}, 42(3):1364--1391, 2013.

\bibitem[New65]{n65}
Donald~J Newman.
\newblock Location of the maximum on unimodal surfaces.
\newblock {\em Journal of the ACM (JACM)}, 12(3):395--398, 1965.

\bibitem[NN89]{nn89}
YE~Nesterov and AS~Nemirovskii.
\newblock Self-concordant functions and polynomial time methods in convex
  programming. preprint, central economic \& mathematical institute, ussr acad.
\newblock {\em Sci. Moscow, USSR}, 1989.

\bibitem[Orl09]{o09}
James~B Orlin.
\newblock A faster strongly polynomial time algorithm for submodular function
  minimization.
\newblock {\em Mathematical Programming}, 118(2):237--251, 2009.

\bibitem[OV20]{ov20}
Neil Olver and L{\'a}szl{\'o}~A V{\'e}gh.
\newblock A simpler and faster strongly polynomial algorithm for generalized
  flow maximization.
\newblock {\em Journal of the ACM (JACM)}, 67(2):1--26, 2020.

\bibitem[Sch87]{s87}
Claus-Peter Schnorr.
\newblock A hierarchy of polynomial time lattice basis reduction algorithms.
\newblock {\em Theoretical computer science}, 53(2-3):201--224, 1987.

\bibitem[Sch98]{s98}
Alexander Schrijver.
\newblock {\em Theory of linear and integer programming}.
\newblock John Wiley \& Sons, 1998.

\bibitem[Sch00]{s00}
Alexander Schrijver.
\newblock A combinatorial algorithm minimizing submodular functions in strongly
  polynomial time.
\newblock {\em Journal of Combinatorial Theory, Series B}, 80(2):346--355,
  2000.

\bibitem[Sch03]{s03}
Alexander Schrijver.
\newblock {\em Combinatorial optimization: polyhedra and efficiency},
  volume~24.
\newblock Springer Science \& Business Media, 2003.

\bibitem[Sho77]{s77}
Naum~Z Shor.
\newblock Cut-off method with space extension in convex programming problems.
\newblock {\em Cybernetics}, 13(1):94--96, 1977.

\bibitem[Sma98]{smale98}
Steve Smale.
\newblock Mathematical problems for the next century.
\newblock {\em The mathematical intelligencer}, 20(2):7--15, 1998.

\bibitem[SV13]{sv13}
Nikhil Srivastava and Roman Vershynin.
\newblock {Covariance estimation for distributions with ${2+\varepsilon}$
  moments}.
\newblock {\em The Annals of Probability}, 41(5):3081 -- 3111, 2013.

\bibitem[Tar86]{t86}
Eva Tardos.
\newblock A strongly polynomial algorithm to solve combinatorial linear
  programs.
\newblock {\em Operations Research}, 34(2):250--256, 1986.

\bibitem[Vai89]{v89}
Pravin~M Vaidya.
\newblock A new algorithm for minimizing convex functions over convex sets.
\newblock In {\em 30th Annual IEEE Symposium on Foundations of Computer Science
  (FOCS)}, pages 338--343, 1989.

\bibitem[V{\'e}g17]{v17}
L{\'a}szl{\'o}~A V{\'e}gh.
\newblock A strongly polynomial algorithm for generalized flow maximization.
\newblock {\em Mathematics of Operations Research}, 42(1):179--211, 2017.

\bibitem[VY96]{vy96}
Stephen~A Vavasis and Yinyu Ye.
\newblock A primal-dual interior point method whose running time depends only
  on the constraint matrix.
\newblock {\em Mathematical Programming}, 74(1):79--120, 1996.

\bibitem[Vyg03]{v03}
Jens Vygen.
\newblock A note on schrijver's submodular function minimization algorithm.
\newblock {\em Journal of Combinatorial Theory, Series B}, 88(2):399--402,
  2003.

\bibitem[YN76]{yn76}
David~B Yudin and Arkadii~S Nemirovski.
\newblock Evaluation of the information complexity of mathematical programming
  problems.
\newblock {\em Ekonomika i Matematicheskie Metody}, 12:128--142, 1976.

\end{thebibliography}

\end{document}